\documentclass[12pt]{article}

\usepackage{amsmath}
\usepackage{amssymb}
\usepackage{amsthm}
\usepackage[pdfpagemode=None]{hyperref}
\usepackage{graphicx}
\usepackage{pstricks}

\newtheorem{thm}{Theorem}[section]

\newtheorem{lem}[thm]{Lemma}

\newtheorem{assumption}[thm]{Assumption}

\newtheorem{pr}[thm]{Proposition}

\newtheorem{definition}[thm]{Definition}

\newtheorem{example}[thm]{Example}
\newenvironment{exmp}{\begin{example}\rm}{\end{example}}
\newtheorem{remark}[thm]{Remark}
\newenvironment{rem}{\begin{remark}\rm}{\end{remark}}
\newtheorem{tab}{Table}

\newcommand{\ord}{{\rm ord}\,}

\def\eps{\varepsilon}

\def\M{{\mathcal M}}
\title{Concavity of Mutual Information Rate for Input-Restricted Finite-State Memoryless
Channels at High SNR}

\author{\begin{tabular}{cc}
Guangyue Han&Brian Marcus\\
University of Hong Kong&University of British Columbia\\
{\em email:} ghan@hku.hk&{\em email:} marcus@math.ubc.ca\\
\end{tabular}}

\date{{\normalsize \today}}

\begin{document}\maketitle\thispagestyle{empty}

\begin{abstract}
We consider a finite-state memoryless channel with i.i.d. channel state and the input Markov process
supported on a mixing finite-type constraint. We discuss the asymptotic behavior of entropy rate of the output hidden Markov chain and deduce that the mutual information rate of such a channel is concave with respect to the parameters of the input Markov processes at high signal-to-noise ratio. In principle, the concavity result enables good numerical approximation of the maximum mutual information rate and capacity of such a channel.
\end{abstract}

\section{Channel Model}

In this paper, we show that for certain input-restricted
finite-state memoryless channels, the mutual information rate, at
high SNR, is effectively a concave function of Markov input
processes of a given order.  While not directly addressed here, the
goal is to help estimate the maximum of this function  and
ultimately  the capacity of such channels (see, for example, the
algorithm of Vontobel, et. al.~\cite{pa04}).

Our  approach depends heavily on results regarding asymptotics and smoothness of entropy rate in special parameterized families of hidden Markov chains, such as those developed in~\cite{jss04},~\cite{or04},~\cite{hm06a},~\cite{hm08}, and
continued here.

We first discuss the nature of the constraints on the input. Let
$\mathcal{X}$ be a finite alphabet. Let $\mathcal{X}^n$ denote																		 
the set of words over $\mathcal{X}$ of length $n$ and let
$\mathcal{X}^{*} = \cup_n \mathcal{X}^n$.  A {\em finite-type}
constraint $\mathcal{S}$ is a subset of $\mathcal{X}^{*}$ defined by
a finite list $\mathcal{F}$ of forbidden words~\cite{lm95, mrs98};
in other words, $\mathcal{S}$ is the set of words over $\mathcal{X}$
that do not contain any element in $\mathcal{F}$ as a contiguous
subsequence. We define $\mathcal{S}_n = \mathcal{S} \cap
\mathcal{X}^n$. The constraint $\mathcal{S}$ is said to be {\em
mixing} if there exists $N$ such that, for any $u, v \in
\mathcal{S}$ and any $n \geq N$, there is a $w \in \mathcal{S}_n$
such that $uwv \in \mathcal{S}$.

In magnetic recording, input sequences are required to satisfy
certain constraints in order to eliminate the most damaging error
events~\cite{mrs98}. The constraints are often mixing finite-type
constraints. The most well-known example is the $(d,k)$-RLL
constraint $\mathcal{S}(d,k)$, which forbids any sequence with fewer
than $d$ or more than $k$ consecutive zeros in between two 1's. For
$\mathcal{S}(d,k)$ with $k < \infty$, a forbidden set $\mathcal{F}$ is:
$$
\mathcal{F}=\{1\underbrace{0\cdots0}_l1: 0 \leq l < d \} \cup \{\underbrace{0\cdots0}_{k+1}\}.
$$
When $k=\infty$, one can choose $\mathcal{F}$ to be
$$
\mathcal{F}=\{1\underbrace{0\cdots0}_l1: 0 \leq l < d \};
$$
in particular when $d=1, k=\infty$, $\mathcal{F}$ can be chosen to be $\{11\}$.

The {\em maximal length} of a forbidden list $\mathcal{F}$ is the
length of the longest word in $\mathcal{F}$.   In general, there can
be many forbidden lists $\mathcal{F}$ which define the same finite
type constraint $\mathcal{S}$. However, we may always choose a list
with smallest maximal length. The {\em (topological) order} of
$\mathcal{S}$ is defined to be $\hat{m}=\hat{m}(\mathcal{S})$ where
$\hat{m}+1$ is the smallest maximal length of any forbidden list
that defines $\mathcal{S}$ (the order of the trivial constraint
$\mathcal{X}^{*}$ is taken to be 0). It is easy to see that the
order of $\mathcal{S}(d,k)$ is $k$ when $k < \infty$, and is $d$ when
$k=\infty$;  $\mathcal{S}(d,k)$ is mixing when $d <k$.

For a stationary stochastic process $X$ over $\mathcal{X}$, the set
of {\em allowed} words with respect to $X$ is defined as
$$
\mathcal{A}(X)=\{w_{-n}^0: n \geq 0, P(X_{-n}^0=w_{-n}^0) > 0\}.
$$
Note that for any $m$-th order stationary Markov process $X$,
the constraint $\mathcal{S} = \mathcal{A}(X)$ is necessarily of
finite-type with order $\hat{m} \le m$, and we say that $X$ is {\em
supported} on $\mathcal{S}$.  Also, $X$ is mixing iff $S$ is mixing
(recall that a Markov chain is mixing if its transition probability matrix, obtained by appropriately enlarging the state space, is irreducible and aperiodic). Note that a Markov chain with support contained in a finite-type
constraint $\mathcal{S}$ may have  order $m < \hat{m}$.

Now, consider a finite-state memoryless channel with finite sets of
channel states $c \in \mathcal{C}$, inputs $x \in \mathcal{X}$,
outputs $z \in \mathcal{Z}$ and input sequences restricted to a
mixing finite-type constraint $\mathcal{S}$. The channel state
process $C$ is assumed to be i.i.d. with $P(C=c)=q_c$. Any stationary input process $X$
must satisfy $\mathcal{A}(X) \subseteq \mathcal{S}$. Let $Z$ denote the stationary
output process corresponding to $X$; then at any time slot, the channel is characterized
by the conditional probability
$$
p(z|x, c) = P(Z=z | X=x, C=c).
$$

We are actually interested in families of channels, as above,
parameterized by $\eps \ge 0$ such that for each $x, c,$ and $z$,
$p(z|x,c)(\eps)$ is an analytic function of $\eps \ge 0$. We
assume that for all $x, c, z$, $p(z|x, c)(\eps)$ is not identically
$0$ as a function of $\eps$, so that for small $\eps >0$, for any input $x$ and channel
state $c$, by analyticity, any output $z$ can occur. We also assume that there is a
one-to-one (not necessarily onto) mapping from $\mathcal{X}$ into $\mathcal{Z}$, $z  =
z(x)$, such that for all $c$ and $x$, $p(z(x)|x,c)(0) = 1$; so,
$\eps$ can be regarded as noise, and $z(x)$ is the noiseless output
corresponding to input $x$.  Note that the output process $Z = Z(X,
\eps)$ depends on the input process $X$ and the parameter value $\eps$; we will often
suppress the notational dependence on $\eps$ or $X$, when it is clear from context.

Prominent examples of such families include input-restricted
versions of the binary symmetric channel with crossover probability $\eps$ (denoted by BSC($\eps)$), the binary
erasure channel with erasure rate $\eps$ (denoted by BEC($\eps)$), and some special Gilbert-Elliott Channels,
where the channel state process is a 2-state i.i.d. process, with one state acting  as BSC($\eps)$ and the other state acting as BSC($k\eps)$ for some fixed $k$; see Section 3 of~\cite{hm08}.

Recall that the {\em entropy rate} of $Z=Z(X,\eps)$ is, as usual, defined as
$$
H(Z)=\lim_{n \to \infty} H_n(Z),
$$
where
$$
H_n(Z)= H(Z_0|Z_{-n}^{-1}) = \sum_{z_{-n}^0} -p(z_{-n}^0) \log
p(z_0|z_{-n}^{-1}).
$$
The {\em mutual information rate} between $Z$ and $X$ can be defined as
$$
I(Z;X)=\lim_{n \to \infty} I_n(Z;X),
$$
where
$$
I_n(Z;X) = H_n(Z) - \frac{1}{n+1}H(Z_{-n}^0|X_{-n}^{0}).
$$
Given the memoryless assumption, one can check that the second term
above is simply $H(Z_0|X_0)$ and in particular does not depend on
$n$.

Under our assumptions, if $X$ is a Markov chain, then for each $\eps
\ge 0$, the output process $Z=Z(X,\eps)$ is a {\em hidden Markov
chain} and in fact satisfies the ``weak Black Hole'' assumption
of~\cite{hm08}, where an asymptotic formula for $H(Z)$ is developed;
the asymptotics are given as an expansion in $\eps$ around $\eps
=0$. In section~\ref{asymptotics}, we further develop these ideas to establish
smoothness properties of $H(Z)$ as a function of $\eps$ and the
Markov chain input $X$ of a fixed order. In particular, we show that
$H(Z)$ can be expressed as $G(X,\eps) + F(X,\eps) \log (\eps)$,
where $G(X,\eps)$ and $F(X,\eps)$ are smooth (i.e., infinitely differentiable) functions of $\eps$
near $0$ for any first order $X$ supported on $\mathcal{S}$ (in fact, $F(X,\eps)$ will
be analytic); the $\log(\eps)$ term arises from the fact that the
support of $X$ will be contained in a non-trivial finite-type
constraint and so $X$ will necessarily have some zero transition
probabilities; this prevents $H(Z)$ from being smooth in $\eps$ at
$0$.

In Section~\ref{CMI}, we apply the smoothness results to show that for a
mixing finite-type constraint $\mathcal{S}$ of order $1$, and sufficiently
small $\eps_0>0$, for each $0 \le \eps \le \eps_0$, $I_n(Z(\eps,X);X)$ and $I(Z(X,\eps);X)$ are
strictly concave on the set of all first order $X$ whose non-zero transition probabilities are not ``too small''.
This will imply that there are unique first order Markov chains
$X_n = X_n(\eps), X_\infty=X_\infty (\eps)$ such that $X_n$
maximizes $I_n(Z(X,\eps),X)$ and $X_\infty$ maximizes
$I(Z(X,\eps),X)$. It will also follow that $X_n(\eps)$ converges
exponentially to $X_\infty(\eps)$ uniformly over $0 \le \eps
\le \eps_0$. In principle, the concavity result enables (via any convex optimization algorithm) good numerical approximation of $X_n(\eps)$ and $X_\infty(\eps)$ and therefore the maximum mutual
information rate over first order $X$. This can be generalized to
$m$-th order Markov chains, and as $m \rightarrow \infty$,
this maximum converges to channel capacity; furthermore it can be generalized to higher
order constraints.

\section{Asymptotics of Entropy Rate}
\label{asymptotics}

\subsection{Key ideas and lemmas}

For simplicity, we consider only mixing finite-type constraints $S$ of
order $1$, and correspondingly only first order
input Markov processes $X$ such that $\mathcal{A}(X) \subseteq S$
(the higher order case is easily reduced to this).
For such $X$ with transition probability matrix $\Pi$,
$(X, C)$ is also a first order Markov chain, with transition probability matrix:
$$
\Omega((x, c), (y, d))=\Pi_{x, y} q_d.
$$
For any $z \in \mathcal{Z}$, define
\begin{equation} \label{omega-z}
\Omega_z((x, c), (y, d))=\Pi_{x, y} q_d p(z|y, d).
\end{equation}
Note that $\Omega_z$ implicitly depends on $\eps$ through $p(z|y, d)$.
One checks that
$$
\sum_{z \in \mathcal{Z}} \Omega_z =\Omega,
$$
and
\begin{equation} \label{p-z}
p(z_{-n}^0)=\pi \Omega_{z_{-n}} \Omega_{z_{-n+1}} \cdots \Omega_{z_{0}} \mathbf{1},
\end{equation}
where $\pi$ is the stationary vector of $\Omega$ and $\mathbf{1}$ is the all $1$'s column vector.

For a given analytic function $f(\eps)$ around $\eps=0$, let $\ord (f(\eps))$ denote its order with respect to $\eps$, i.e., the degree of the first non-zero term of its Taylor series expansion around $\eps=0$. Thus, the orders $\ord(p(z|x,c))$ determine the orders
$\ord(p(z_{-n}^0))$ and similarly orders of conditional probabilities $\ord(p(z_0|z_{-n}^{-1}))$.

\begin{exmp} \label{rll}
Consider a binary symmetric channel with crossover probability $\eps$ and a binary input Markov chain $X$ supported on the $(1, \infty)$-RLL constraint with transition probability matrix
$$
\Pi=\left[\begin{array}{cc}
1-p&p\\
1&0\\
\end{array}\right],
$$
where $0 < p < 1$. Here there is only one channel state, and so we can suppress dependence on the channel state.
The channel is characterized by the conditional probability
$$
p(z|x) = p(z|x)(\eps)=
\left\{
\begin{array}{cc}
1-\eps & \mbox{ if } z = x \\
\eps & \mbox{ if } z \ne x
\end{array}
\right.
$$
Let $Z$ be the corresponding output binary hidden Markov chain.
Now we have
$$
\Omega_0=\left[\begin{array}{cc}
(1-p)(1-\eps)&p \eps\\
1-\eps&0\\
\end{array}\right], \Omega_1=\left[\begin{array}{cc}
(1-p) \eps&p (1-\eps)\\
\eps&0\\
\end{array}\right].
$$
The stationary vector $\pi=(1/(p+1),p/(p+1))$, and one computes, for instance,
$$
p(z_{-2} z_{-1} z_0 = 110) =
\pi \Omega_1 \Omega_1 \Omega_0 \mathbf{1} = \frac{2p - p^2}{1+p}\eps +O(\eps^2),
$$
which has order 1.
\end{exmp}

Let $\mathcal{M}$ denote the set of all first order stationary Markov chains $X$ satisfying $\mathcal{A}(X) \subseteq \mathcal{S}$. Let $\mathcal{M}_{\delta}$, $\delta \geq 0$, denote the set of all $X \in \mathcal{M}$ such that $p(w_{-1}^0) > \delta$ for all $w_{-1}^0 \in \mathcal{S}_{2}$. Note that whenever $X \in  \mathcal{M}_{0}$, i.e., $\mathcal{A}(X)=\mathcal{S}$, $X$ is mixing (thus its transition probability matrix $\Pi$ is primitive) since $S$ is mixing, so $X$ is completely determined by its transition probability matrix $\Pi$. For the purpose of this paper, however, we find it convenient to identify each $X \in \mathcal{M}_0$ with its vector of {\em joint} probabilities $\vec{p}=\vec{p}_X$ on words of length 2 instead:
$$
\vec{p}=\vec{p}_X=(P(X_{-1}^0=w_{-1}^0): w_{-1}^0 \in \mathcal{S}_{2});
$$
sometimes we write $X = X(\vec{p})$.

In the following, for any parameterized sequence of functions $f_{n, \lambda}(\eps)$ ($\eps$ is real or complex), we use
$$
f_{n, \lambda}(\eps) = \hat{O}(\eps^n) \mbox{ on } \Lambda
$$
to mean that there exist constants $C, \beta_1, \beta_2 > 0$, $\eps_0 > 0$ such that for all $n$, all $\lambda \in \Lambda$ and all $0 \leq |\eps| \leq \eps_0$,
$$
|f_{n, \lambda}(\eps)| \leq n ^{\beta_1} (C |\eps|^{\beta_2})^n.
$$
Note that $f_{n, \lambda}(\eps) =\hat{O}(\eps^n)$ on $\Lambda$ implies that there exists $\eps_0 > 0$ and $0 < \rho < 1$ such that $|f_{n, \lambda}(\eps)| < \rho^n$ for all $|\eps| \leq \eps_0$, all $\lambda \in \Lambda$ and large enough $n$.
One also checks that a $\hat{O}(\eps^n)$-term is unaffected by multiplication of an exponential function (thus polynomial function) in $n$ and a polynomial function in $1/\eps$;
\begin{rem} \label{TheSameEps}
For any given $f_{n, \lambda}(\eps)= \hat{O}(\eps^n)$, there exists $\eps_0 > 0$ and $0 < \rho < 1$ such that $|g_1(n) g_2(1/\eps) f_{n, \lambda}(\eps)| \leq \rho^n$, for all $|\eps| \leq \eps_0$, all $\lambda \in \Lambda$, all polynomial functions $g_1(n), g_2(1/\eps)$ and large enough $n$.
\end{rem}

Of course, the output joint probabilities $p(z_{-n}^0)$ and conditional probabilities $p(z_0|z_{-n}^{-1})$
implicitly depend on $\vec{p} \in \mathcal{M}_0$ and $\eps$. The following result asserts that for small $\eps$,
the total probability of output sequences with ``large'' order is exponentially small, uniformly over all input processes.

\begin{lem}   \label{MATH2603}
For any fixed $0 < \alpha < 1$,
$$
\sum_{\ord(p(z_{-n}^{-1})) \geq \alpha n} p(z_{-n}^{-1}) = \hat{O}(\eps^n) \mbox{ on } \mathcal{M}_0.
$$
\end{lem}

\begin{proof}
Note that for any hidden Markov chain sequence $z_{-n}^{-1}$, we have
\begin{equation}  \label{40minutes}
p(z_{-n}^{-1})=\sum p(x_{-n}^{-1}, c_{-n}^{-1}) \prod_{i=-n}^{-1} p(z_i|x_i, c_i),
\end{equation}
where the summation is over all $(x_{-n}^{-1}, c_{-n}^{-1})$. Now consider $z_{-n}^{-1}$ with $k=\ord(p(z_{-n}^{-1})) \geq \alpha n$. One checks
that for $\eps$ small enough there exists a positive constant $C$ such that $p(z|x, c) \leq C \eps$ for $(x, c, z)$ with $\ord(p(z|x, c)) \geq 1$, and thus the term $\prod_{i=-n}^{-1} p(z_i|x_i, c_i)$ as in (\ref{40minutes}) is upper bounded by $C^{k} \eps^{k}$, which is upper bounded by $C^{\alpha n} \eps^{\alpha n}$ for $\eps < 1/C$. Noticing that $\sum_{x_{-n}^{-1}, c_{-n}^{-1}} p(x_{-n}^{-1}, c_{-n}^{-1})=1$, we then have, for $\eps$ small enough,
$$
\sum_{\ord(p(z_{-n}^{-1})) \geq \alpha n} p(z_{-n}^{-1}) \leq \sum_{z_{-n}^{-1}} \sum_{x_{-n}^{-1}, c_{-n}^{-1}} p(x_{-n}^{-1}, c_{-n}^{-1}) C^{\alpha n} \eps^{\alpha n} \leq |\mathcal{Z}|^n C^{\alpha n} \eps^{\alpha n},
$$
which immediately implies the lemma.
\end{proof}

Now for any $\delta > 0$, consider a first order Markov chain $X \in \mathcal{M}_{\delta}$ with transition probability matrix $\Pi$ (note that $X$ is necessarily mixing). Let $\Pi^{\mathbb{C}}$ denote a complex ``transition probability matrix'' obtained by perturbing all entries of $\Pi$ to complex numbers, while satisfying $\sum_y \Pi^{\mathbb{C}}_{xy}=1$. Then through solving the following system of equations
$$
\pi^{\mathbb{C}} \Pi^{\mathbb{C}}=\pi^{\mathbb{C}}, \qquad \sum_y \pi^{\mathbb{C}}=1,
$$
one can obtain a complex ``stationary probability'' $\pi^{\mathbb{C}}$, which is uniquely defined if the perturbation of $\Pi$ is small enough. It then follows that under a complex perturbation of $\Pi$, for any Markov chain sequence $x_{-n}^0$, one can obtain a complex version of $p(x_{-n}^0)$ through complexifying all terms in the following expression:
$$
p(x_{-n}^0)=\pi_{x_{-n}} \Pi_{x_{-n}, x_{-n+1}} \cdots \Pi_{x_{-1}, x_0},
$$
namely,
$$
p^{\mathbb{C}}(x_{-n}^0)=\pi^{\mathbb{C}}_{x_{-n}} \Pi^{\mathbb{C}}_{x_{-n}, x_{-n+1}} \cdots \Pi^{\mathbb{C}}_{x_{-1}, x_0};
$$
in particular, the joint probability vector $\vec{p}$ can be complexified to $\vec{p}^{\mathbb{C}}$ as well. We then use $\mathcal{M}_{\delta}^{\mathbb{C}}(\eta)$, $\eta > 0$, to denote the $\eta$-perturbed complex version of $\mathcal{M}_{\delta}$; more precisely,
$$
\mathcal{M}_{\delta}^{\mathbb{C}}(\eta)=\{(\vec{p}^{\mathbb{C}}(w_{-1}^0): w_{-1}^0 \in \mathcal{S}_2)| \; \|\vec{p}^{\mathbb{C}} - \vec{p}\| \leq \eta \mbox{ for some } \vec{p} \in \mathcal{M}_{\delta}\},
$$
which is well-defined if $\eta$ is small enough. Furthermore, together with a small complex perturbation of $\eps$, one can obtain a well-defined complex version $p^{\mathbb{C}}(z_{-n}^0)$ of $p(z_{-n}^0)$ through complexifying (\ref{omega-z}) and (\ref{p-z}).

Using the same argument as in Lemma~\ref{MATH2603} and applying the triangle inequality to the absolute value of (\ref{40minutes}), we have
\begin{lem} \label{complex}
For any $\delta > 0$, there exists $\eta > 0$ such that for any fixed $0 < \alpha < 1$,
$$
\sum_{\ord(p^{\mathbb{C}}(z_{-n}^{-1})) \geq \alpha n} |p^{\mathbb{C}}(z_{-n}^{-1})| = \hat{O}(|\eps|^n) \mbox{ on } \mathcal{M}_{\delta}^{\mathbb{C}}(\eta).
$$
\end{lem}

By Lemma~\ref{MATH2603} and Lemma~\ref{complex}, we can focus our attention on output sequences with relatively small order. For a fixed positive $\alpha$, a sequence $z_{-n}^{-1} \in \mathcal{Z}^n$ is said to be {\em $\alpha$-typical} if $\ord(p(z_{-n}^{-1})) \leq \alpha n$; let $T_{n}^{\alpha}$ denote the set of all $\alpha$-typical $\mathcal{Z}$-sequences with length $n$.  Note that this definition is independent of $\vec{p} \in \mathcal{M}_0$.

For a smooth mapping  $f(\vec{x})$ from $\mathbb{R}^k$ to $\mathbb{R}$
and a nonnegative integer $\ell$, $D^\ell_{\vec{x}}f$ denotes the $\ell$-th total
derivative  with respect to $\vec{x}$; for instance,
$$
D_{\vec{x}}f=\left(\frac{\partial f}{\partial x_i}\right)_{i} \mbox{ and } D^2_{\vec{x}}f=\left(\frac{\partial^2 f}{\partial x_i \partial x_j}\right)_{i, j}.
$$
In particular, if $\vec{x} = \vec{p} \in \mathcal{M}_0$ or $\vec{x} = (\vec{p},\eps) \in \mathcal{M}_0 \times [0,1]$, this defines the derivatives $D_{\vec{p}}^{l}p(z_0|z_{-n}^{-1})$ or $D_{\vec{p},\eps}^{l}p(z_0|z_{-n}^{-1})$.
We shall use $| \cdot |$ to denote the Euclidean norm (of a vector or a matrix), and we shall use $\| A \|$ to denote the norm of a matrix $A$ as a linear map under the Euclidean norm, i.e.,
$$
\|A\|=\sup_{x \neq \vec{0}} \frac{|Ax|}{|x|}.
$$
It is well known that $\| A \| \leq |A|$.

In this paper, we are interested in functions of $\vec{q}=(\vec{p}, \eps)$. For any smooth function $f$ of $\vec{q}$ and
$\vec{n}=(n_1, n_2, \cdots, n_{|\mathcal{S}_2|+1}) \in \mathbb{Z}_+^{|\mathcal{S}_2|+1}$, define
$$
f^{(\vec{n})}=\frac{\partial^{|\vec{n}|} f}{\partial q_1^{n_1}
\partial q_2^{n_2} \cdots \partial q_{|\mathcal{S}_2|+1}^{n_{|\mathcal{S}_2|+1}}},
$$
here $|\vec{n}|$ denotes the order of the $\vec{n}$-th derivative of
$f$ with respect to $\vec{q}$, and is defined as
$$
|\vec{n}|=n_1+n_2+\cdots+n_{|\mathcal{S}_2|+1}.
$$

The next result shows, in a precise form, that for $\alpha$-typical sequences  $z_{-n}^0$, the  derivatives, of all orders, of the difference between $p(z_0|z_{-n}^{-1})$ and $p(z_0|z_{-n-1}^{-1})$ converge exponentially in $n$, uniformly in $\vec{p}$ and $\eps$. For $n \leq m , \hat{m} \leq 2n$, define
$$
T_{n, m, \hat{m}}^{\alpha}=\{(z_{-m}^0, \hat{z}_{-\hat{m}}^0) \in \mathcal{Z}^{m+1} \times \mathcal{Z}^{\hat{m}+1}| z_{-n}^{-1}=\hat{z}_{-n}^{-1} \mbox{ is } \mbox{$\alpha$-typical}\}.
$$

\begin{pr} \label{overwhelm}
Assume $n \leq m , \hat{m} \leq 2n$. Given $\delta_0 >0$, there exists $\alpha > 0$ such that for any $\ell$
$$
|D_{\vec{p},\eps}^{\ell}p(z_0|z_{-m}^{-1})-D_{\vec{p},\eps}^{\ell} p(z_0|z_{-\hat{m}}^{-1})| = \hat{O}(\eps^n) \mbox{ on } \mathcal{M}_{\delta_0} \times T_{n, m, \hat{m}}^{\alpha}.
$$
\end{pr}

The proof of Proposition~\ref{overwhelm} depends on estimates of derivatives of certain induced maps
on a simplex, which we now describe. Let $\mathcal{W}$ denote the unit simplex in $\mathbb{R}^{|\mathcal{X}| \cdot |\mathcal{C}|}$, i.e., the set of nonnegative vectors, which sum to $1$, indexed by the joint input-state space $\mathcal{X} \times \mathcal{C}$. For any $z \in \mathcal{Z}$, $\Omega_z$ induces a mapping $f_z$ defined on $\mathcal{W}$ by
\begin{equation} \label{induce}
f_z(w)=\frac{w \Omega_z}{w \Omega_z \mathbf{1}}.
\end{equation}
Note that $\Omega_z$ implicitly depends on the input Markov chain $\vec{p} \in \mathcal{M}_0$ and $\eps$, and thus so does $f_z$. While $w \Omega_z \mathbf{1}$ can vanish at $\eps =0$, it is easy to check that
for all $w \in \mathcal{W}$, $\lim_{\eps \rightarrow 0} f_z(w)$ exists, and so $f_z$ can be defined at $\eps=0$.
Let $O_M$ denote the largest order of all entries of $\Omega_z$ (with respect to $\eps$) for all $z \in \mathcal{Z}$, or equivalently, the largest order of $p(z|x, c)(\eps)$ over all possible $x, c, z$.

For $\eps_0, \delta_0 > 0$, let
$$
U_{\delta_0, \eps_0} = \{\vec{p} \in
\mathcal{M}_{\delta_0}, \eps \in [0, \eps_0] \}.
$$

\begin{lem}  \label{Lipschitz}
Given $\delta_0 >0$, there exists $\eps_0 >0$ and $C_e>0$ such that
on $U_{\delta_0, \eps_0}$ for all $z \in \mathcal{Z}$, $|D_wf_z| \le C_e/\eps^{2O_M}$ on the entire simplex $\mathcal{W}$.
\end{lem}

\begin{proof}
Given $\delta_0 > 0$, there exist $\eps_0 > 0$ and $C > 0$ such that for any $z \in \mathcal{Z}$, $w \in \mathcal{W}$, we have, for all $0 \leq \eps \leq \eps_0$,
$$
|w \Omega_z \mathbf{1}| \geq C \eps^{O_M}.
$$
We then apply the quotient rule to establish the lemma.
\end{proof}

For any sequence $z_{-N}^{-1} \in \mathcal{Z}^N$, define
$$
\Omega_{z_{-N}^{-1}} \stackrel{\triangle}{=} \Omega_{z_{-N}} \Omega_{z_{-N+1}} \cdots \Omega_{z_{-1}}.
$$
Similar to (\ref{induce}), $\Omega_{z_{-N}^{-1}}$ induces a mapping $f_{z_{-N}^{-1}}$ on $\mathcal{W}$ by:
$$
f_{z_{-N}^{-1}}(w)=\frac{w \Omega_{z_{-N}^{-1}}}{w \Omega_{z_{-N}^{-1}} \mathbf{1}}.
$$
By the chain rule, Lemma~\ref{Lipschitz} gives upper bounds on derivatives of $f_{z_{-N}^{-1}}$.  However, these bounds can be improved considerably in certain cases, as we now describe.
A sequence $z_{-N}^{-1} \in \mathcal{Z}^N$ is {\em $Z$-allowed} if there exists $x_{-N}^{-1} \in \mathcal{A}(X)$ such that
$$
z_{-N}^{-1}=z(x_{-N}^{-1}) \stackrel{\triangle}{=} (z(x_{-N}), z(x_{-N+1}), \cdots, z(x_{-1})).
$$
Note that $z_{-N}^{-1}$ is $Z$-allowed iff $\ord(p(z_{-N}^{-1}))=0$.

Since $\Pi$ is a primitive matrix, there exists a positive integer $e$ such that $\Pi^e > 0$. For any $z \in \mathcal{Z}$, let $I_z$ denote the set of indices of the columns $(x, c)$ of $\Omega_z$ such that $z=z(x)$; note that $I_z$ can be empty for some $z \in \mathcal{Z}$.

\begin{lem}  \label{HAY}
Assume that $X \in \mathcal{M}_0$. For any $Z$-allowed sequence $z_{-N}^{-1}=z(x_{-N}^{-1}) \in \mathcal{Z}^N$ (here $x_{-N}^{-1} \in \mathcal{S}$), if $N \geq 2eO_M$, we have
$$
\ord((\Omega_{z_{-N}^{-1}})(s, t_1))=\ord((\Omega_{z_{-N}^{-1}})(s, t_2)),
$$
for all $s$, and any $t_1, t_2 \in I_{z_{-1}}$, and
$$
\ord((\Omega_{z_{-N}^{-1}})(s, t_1)) < \ord((\Omega_{z_{-N}^{-1}})(s, t_2)),
$$
for all $s$, and any $t_1 \in I_{z_{-1}}$, $t_2 \not \in I_{z_{-1}}$.
\end{lem}

\begin{proof}

Let $s=(\hat{x}_{-N-1}, \hat{c}_{-N-1})$, $t=(\hat{x}_{-1}, \hat{c}_{-1})$ $\in \mathcal{X} \times \mathcal{C}$. Then
$$
\Omega_{z_{-N}^{-1}}(s, t)=P((X_{-1}, C_{-1})=(\hat{x}_{-1}, \hat{c}_{-1}), Z_{-N}^{-1}=z_{-N}^{-1}|(X_{-N-1}, C_{-N-1})=(\hat{x}_{-N-1}, \hat{c}_{-N-1}))
$$
$$
=p((\hat{x}_{-1}, \hat{c}_{-1}), z_{-N}^{-1}|(\hat{x}_{-N-1}, \hat{c}_{-N-1})).
$$
It then follows that
$$
\ord(\Omega_{z_{-N}^{-1}}(s, t))=\ord(p((\hat{x}_{-1}, \hat{c}_{-1}), z_{-N}^{-1}|(\hat{x}_{-N-1}, \hat{c}_{-N-1})))=\ord(p((\hat{x}_{-N-1}, \hat{c}_{-N-1}), z_{-N}^{-1}, (\hat{x}_{-1}, \hat{c}_{-1}))).
$$
Since
$$
p((\hat{x}_{-N-1}, \hat{c}_{-N-1}), z_{-N}^{-1}, (\hat{x}_{-1}, \hat{c}_{-1}))=\sum_{\hat{x}_{-N}^{-2}, \hat{c}_{-N}^{-2}} p(\hat{x}_{-N-1}^{-1}, \hat{c}_{-N-1}^{-1}, z_{-N}^{-1}),
$$
we have
$$
\ord(\Omega_{z_{-N}^{-1}}(s, t))=\min \sum_{i=-N}^{-1} \ord(p(z_i|\hat{x}_i, \hat{c}_i)),
$$
where the minimization is over all sequences $(\hat{x}_{-N}^{-2}, \hat{c}_{-N}^{-2})$ such that $\hat{x}_{-N-1}^{-1} \in \mathcal{S}$.

Since $\Pi^e > 0$, there exists some $\hat{x}_{-N}^{-N-1+e}$ such that $\hat{x}_{-N-1+e}=x_{-N-1+e}$ and $p(\hat{x}_{-N-1}^{-N-1+e}) > 0$, and there exists some $\hat{x}_{-e}^{-2}$ such that $\hat{x}_{-e}=x_{-e}$ and $p(\hat{x}_{-e}^{-1}) > 0$. It then follows from $\ord(p(z|x, c)) \leq O_M$ that, as long as $N \geq 2eO_M$, for any fixed $t$ and any choice of order minimizing sequence $(\hat{x}_{-N}^{-2}(t), \hat{c}_{-N}^{-2}(t))$, there exist $0 \leq i_0=i_0(t), j_0=j_0(t) \leq e O_M$ such that $z(\hat{x}_i^j(t))=z_i^j$ if and only if $i \geq -N-1+i_0(t)$ and $j \leq -1-j_0(t)$. One further checks that, for any choice of order minimizing sequences corresponding to $t$, $(\hat{x}_{-N}^{-2}(t), \hat{c}_{-N}^{-2}(t))$,
$$
\sum_{i=-N}^{i_0(t)} \ord(p(z_i|\hat{x}_i(t), \hat{c}_i(t))),
$$
does not depend on $t$, whereas $j_0(t)=0$ if and only if $z(\hat{x}_{-1})=z_{-1}$. This immediately implies the lemma.

\end{proof}

\begin{exmp}(continuation of Example~\ref{rll})

Recall that
$$
\Omega_0=\left[\begin{array}{cc}
(1-p)(1-\eps)&p \eps\\
1-\eps&0\\
\end{array}\right],  \qquad \Omega_1=\left[\begin{array}{cc}
(1-p) \eps&p (1-\eps)\\
\eps&0\\
\end{array}\right].
$$
First, observe that the only $Z$-allowed sequences are $00, 01, 10$; then straightforward computations show that
\begin{eqnarray}
\notag \Omega_0 \Omega_0&=&\left[ \begin{array}{cc}
(1-p)^2(1-\eps)^2+p\eps(1-\eps)&p(1-p)\eps(1-\eps)\\
(1-p) (1-\eps)^2&p\eps(1-\eps)\\
\end{array} \right],\\
\notag \Omega_0 \Omega_1&=&\left[ \begin{array}{cc}
(1-p)^2 \eps (1-\eps)+p \eps^2&p(1-p)(1-\eps)^2\\
(1-p) \eps (1-\eps) & p (1-\eps)^2\\
\end{array} \right],\\
\notag \Omega_1 \Omega_0&=&\left[ \begin{array}{cc}
(1-p)^2 \eps (1-\eps)+p(1-\eps)^2&p(1-p)\eps^2\\
(1-p) \eps (1-\eps) & p \eps^2\\
\end{array} \right].
\end{eqnarray}
Note that in the spirit of Lemma~\ref{HAY}, for each of these three matrices, there is a unique column, each of whose entries minimizes the orders over all the entries in the same row.
\end{exmp}

Now fix $N \geq 2 e O_M$. For any $w \in \mathcal{W}$, let $v=f_{z_{-N}^{-1}}(w)$. Note that the mapping $f_{z_{-N}^{-1}}$ implicitly depends on $\eps$, so $v$ is in fact a function of $\eps$. If $z_{-N}^{-1}=z(x_{-N}^{-1}) \in \mathcal{Z}^N$ is $Z$-allowed, by Lemma~\ref{HAY}, when $\eps=0$,
\begin{itemize}
\item $v_i=0$ if and only if $i \not \in I_{z_{-1}}$,
\item for each $i=(x_{-1}, c_{-1}) \in I_{z_{-1}}$, $v_i=q_{c_{-1}}$, which does not depend on $w$.
\end{itemize}
Let $q(z) \in \mathcal{W}$ be the point defined by $q(z)_{(x,c)} = q_c$ for all $(x, c)$ with $z(x)=z$ and $0$ otherwise. If $z_{-N}^{-1}$ is $Z$-allowed, then
$$
\lim_{\eps \to 0} f_{z_{-N}^{-1}}(w)=q(z_{-1});
$$
thus, in this limiting sense, at $\eps=0$, $f_{z_{-N}^{-1}}$ maps the entire simplex $\mathcal{W}$ to a single point $q(z_{-1})$. The following lemma says that if $z_{-N-1}^{-1}$ is $Z$-allowed, then in a small neighbourhood of $q(z_{-N-1})$, the derivative of $f_{z_{-N}^{-1}}$ is much smaller than what would be given by repeated application of Lemma~\ref{Lipschitz}.

\begin{lem}  \label{Lipschitz2}
Given $\delta_0 >0$, there exists $\eps_0 >0$  and $C_c>0$ such that on $U_{\delta_0, \eps_0}$,
if $z_{-N-1}^{-1}$ is $Z$-allowed, then $|D_wf_{z_{-N}^{-1}}| \le C_c \eps$ on some neighbourhood of $q(z_{-N-1})$.
\end{lem}

\begin{proof}

By the observations above, for all $w \in \mathcal{W}$, we have
$$
f_{z_{-N}^{-1}}(w)=q(z_{-1})+\eps r(w),
$$
where $r(w)$ is a rational vector-valued function with common denominator of order $0$ (in $\eps$) and leading coefficient uniformly bounded away from $0$ near $w=q(z_{-N-1})$ over all $\vec{p} \in \mathcal{M}_{\delta_0}$.
The lemma then immediately follows.
\end{proof}

\subsection{Proof of Proposition~\ref{overwhelm}}

We now explain the rough idea of the proof of Proposition~\ref{overwhelm}, for only the special case $\ell
=0$, i.e., exponential convergence of the difference between $p(z_0|z_{-n}^{-1})$ and $p(z_0|z_{-n-1}^{-1})$.
Let $N$ be as above and for simplicity consider only output sequences of length a multiple $N$: $n = n_0 N$. We can compute an estimate of $D_wf_{z_{-n}^0}$ by using the chain rule (with appropriate
care at $\eps =0$) and multiplying the estimates on $|D_wf_{z_{-iN}^{(-i+1)N}}|$ given by Lemmas~\ref{Lipschitz}
and ~\ref{Lipschitz2}. This yields an estimate of the form, $|D_wf_{z_{-n}^0}| \le (A\eps^{1-B\alpha})^n$ for some
constants $A$ and $B$, on the entire simplex $\mathcal{W}$. If $\alpha$ is sufficiently small and $z_{-n}^{-1}$ is
$\alpha$-typical, then the estimate from Lemma~\ref{Lipschitz2} applies enough of the time that $f_{z_{-n}^0}$ exponentially contracts the simplex. Then, interpreting elements of the simplex as conditional probabilities $p((x_i,c_i) = \cdot | z_{-m}^i)$, we obtain exponential convergence of the difference $|p(z_0|z_{-n}^{-1})-p(z_0|z_{-n-1}^{-1})|$, as desired.

\begin{proof}[Proof of Proposition~\ref{overwhelm}]

For simplicity, we only consider the special case that $n=n_0
N, m=m_0 N, \hat{m}=\hat{m}_0 N$ for a fixed $N \geq 2 e O_M$;
the general case can be easily reduced to this special case.
For the sequences $z_{-m}^{-1}, \hat{z}_{-\hat{m}}^{-1}$,
define their ``blocked'' version $[z]_{-m_0}^{-1},
[\hat{z}]_{-\hat{m}_0}^{-1}$ by setting
$$
[z]_i=z_{i N}^{(i+1)N-1}, i=-m_0, -m_0+1, \cdots, -1,
\qquad [\hat{z}]_j=\hat{z}_{j N}^{(j+1)N-1}, j=-\hat{m}_0, -\hat{m}_0+1, \cdots, -1.
$$

Let
$$
w_{i, -m}=w_{i, -m}(z^i_{-m})=p((x_i, c_i)=\cdot \;|z^i_{-m}),
$$
where $\cdot$ denotes the possible states of Markov chain $(X, C)$. Then one
checks that
\begin{equation}   \label{negative1}
p(z_0|z_{-m}^{-1})=w_{-1, -m} \Omega_{z_0} \mathbf{1}
\end{equation}
and $w_{i, -m}$ satisfies the following iteration
$$
w_{(i+1), -m}=f_{z_{i+1}}(w_{i, -m}) \qquad -n \leq i \leq -1,
$$
and the following iteration (corresponding to the blocked chain $[z]_{-m_0}^{-1}$)
\begin{equation}  \label{iteration}
w_{(i+1)N-1, -m}=f_{[z]_i}(w_{i N-1, -m}) \qquad -n_0 \leq i \leq -1,
\end{equation}
starting with
$$
w_{-n-1, -m} = p((x_{-n-1}, c_{-n-1})=\cdot|z_{-m}^{-n-1}).
$$
Similarly let
$$
\hat{w}_{i, -\hat{m}}=\hat{w}_{i, -\hat{m}}(\hat{z}^i_{-\hat{m}})=p((x_i, c_i)=\cdot \;|\hat{z}^i_{-\hat{m}}),
$$
which also satisfies the same iterations as above, however starting with
$$
\hat{w}_{-n-1, -\hat{m}} = p((x_{-n-1}, c_{-n-1})=\cdot|\hat{z}_{-\hat{m}}^{-n-1}).
$$

We say $[z]_{-n_0}^{-1}$  ``continues'' between $[z]_{i-1}$ and $[z]_i$ if $[z]_{i-1}^i$ is $Z$-allowed; on the other hand, we say $[z]_{-n_0}^{-1}$ ``breaks'' between $[z]_{i-1}$ and $[z]_i$ if it does not continue between
$[z]_{i-1}$ and $[z]_i$, namely, if one of the following occurs
\begin{enumerate}
\item $[z]_{i-1}$ is not $Z$-allowed;
\item $[z]_i$ is not $Z$-allowed;
\item both $[z]_{i-1}$ and $[z]_i$ are $Z$-allowed, however $[z]_{i-1}^i$ is not $Z$-allowed.
\end{enumerate}

Iteratively applying Lemma~\ref{Lipschitz}, there is a positive constant $C_e$ such that
\begin{equation}  \label{break}
|D_w f_{[z]_i}| \leq C_e^N/\eps^{2NO_M},
\end{equation}
on the entire simplex $\mathcal{W}$. In particular, this holds when $[z]_{-n_0}^{-1}$ ``breaks''
between $[z]_{i-1}$ and $[z]_i$. When $[z]_{-n_0}^{-1}$ ``continues'' between $[z]_{i-1}$ and
$[z]_i$, by Lemma~\ref{Lipschitz2}, we have that if
$\eps$ is small enough, there is a constant $C_c > 0$ such that
\begin{equation} \label{continue}
|D_w f_{[z]_i}| \leq C_c \eps
\end{equation}
on $f_{[z]_{i-1}}(\mathcal{W})$.

Now, apply the mean value theorem, we deduce that there exist $\xi_i$, $-n_0 \leq i \leq -1$ (here $\xi_i$ is a convex combination of $w_{-i N-1, -m}$ and $\hat{w}_{-i N -1, -\hat{m}}$) such that
$$
|w_{-1, -m}-\hat{w}_{-1, -\hat{m}}|=|f_{[z]_{-n_0}^{-1}}(w_{-n_0 N-1, -m})-f_{[z]_{-n_0}^{-1}}(\hat{w}_{-n_0 N-1, -\hat{m}})|
$$
$$
\leq \prod_{i=-n_0}^{-1} \|D_w f_{[z]_i}(\xi_i)\| |w_{-n_0 N-1, -m}-\hat{w}_{-n_0 N-1, -\hat{m}}|.
$$
Since $z_{-n}^{-1}$ is $\alpha$-typical, $[z]_{-n_0}^{-1}$ breaks at most $3\alpha n$ times; in other words, there are at least $(1/N-3\alpha)n$ $i$'s corresponding to (\ref{continue}) and at most $3\alpha n$ $i$'s corresponding to (\ref{break}). We then have
\begin{equation} \label{CcCe}
\prod_{i=-n_0}^{-1} \|D_w f_{[z]_i}(\xi_i)\| \leq C_c^{(1/N-3\alpha)n} C_e^{3\alpha N n} \eps^{(1/N-3\alpha-6 N O_M \alpha)n}.
\end{equation}
Let $\alpha_0=1/(N(3+6NO_M))$. Evidently, when $\alpha < \alpha_0$, $1/N-3\alpha-6 N O_M \alpha$ is strictly positive, we then have
\begin{equation}  \label{ww}
|w_{-1, -m}-\hat{w}_{-1, -\hat{m}}| = \hat{O}(\eps^n) \mbox{ on } \mathcal{M}_{\delta_0} \times T_{n, m, \hat{m}}^{\alpha}.
\end{equation}
It then follows from (\ref{negative1}) that
$$
|p(z_0|z_{-m}^{-1})-p(\hat{z}_0|\hat{z}_{-\hat{m}}^{-1})| = \hat{O}(\eps^n) \mbox{ on } \mathcal{M}_{\delta_0} \times T_{n, m, \hat{m}}^{\alpha}.
$$

We next show that for each $\vec{k}$, there is a positive constant $C_{|\vec{k}|}$ such that
\begin{equation} \label{prelim}
|w_{i, -m}^{(\vec{k})}|, |\hat{w}_{i, -\hat{m}}^{(\vec{k})}| \leq n^{|\vec{k}|}C_{|\vec{k}|}/\eps^{|\vec{k}|};
\end{equation}
here, the superscript ${}^{(\vec{k})}$ denotes the $\vec{k}$-th order derivative with respect to $\vec{q}=(\vec{p}, \eps)$.
In fact, the partial derivatives with respect to $\vec{p}$ are upper bounded in norm by
$n^{|\vec{k}|}C_{|\vec{k}|}$.

To illustrate the idea, we first prove (\ref{prelim}) for $|\vec{k}|=1$. Recall that
$$
w_{i, -m}=p((x_i, c_i)=\cdot|z_{-m}^i)=\frac{p((x_i, c_i)=\cdot, z_{-m}^i)}{p(z_{-m}^i)}.
$$
Let $q$ be a component of $\vec{q}=(\vec{p}, \eps)$. Then,
$$
\left|\frac{\partial}{\partial q} \left(\frac{p((x_i, c_i), z_{-m}^i)}{p(z_{-m}^i)} \right) \right| =
\left| \frac{p((x_i, c_i), z_{-m}^i)}{p(z_{-m}^i)}
\left(\frac{\frac{\partial}{\partial q} p((x_i, c_i), z_{-m}^i)}{p((x_i, c_i), z_{-m}^i)}   -
\frac{\frac{\partial}{\partial q}  p(z_{-m}^i)}   {p(z_{-m}^i)}
\right)
\right|
$$
$$
\le
\left| \frac{p((x_i, c_i)=\cdot, z_{-m}^i)}{p(z_{-m}^i)} \right|
\left(\left| \frac{\frac{\partial}{\partial q} p((x_i, c_i)=\cdot, z_{-m}^i)}   {p((x_i, c_i)=\cdot, z_{-m}^i)} \right|   +
\left|\frac{\frac{\partial}{\partial q} p(z_{-m}^i)}   {p(z_{-m}^i)} \right|
\right).
$$
We first consider the partial derivative with respect to $\eps$, i.e., $q=\eps$.
Since the first factor is bounded above by $1$, it suffices to show that both terms of the second factor are  $mO(1/\eps)$ (applying the argument to both $z_{-m}^i$ and  $\hat{z}_{-\hat{m}}^i$ and recalling that $n \le m, \hat{m} \le 2n$). We will prove this only for $\left|\frac{\partial}{\partial \eps} p(z_{-m}^i)/p(z_{-m}^i)\right|$,
with the proof for the other term being similar.
Now
\begin{equation}
\label{sum}
p(z_{-m}^i)=\sum g(x_{-m}^{-1}, c_{-m}^{-1}),
\end{equation}
where
$$
g(x_{-m}^{-1}, c_{-m}^{-1}) =
p(x_{-m}) \prod_{j=-m}^{i-1} p(x_{j+1}|x_j) \prod_{j=-m}^i p(c_j) \prod_{j=-m}^i p(z_j|x_j, c_j)
$$
and the summation is over all Markov chain sequences $x_{-m}^i$ and channel state sequences $c_{-m}^i$.
Clearly, $\frac{\partial}{\partial \eps} p(z_j|x_j, c_j)/p(z_j|x_j, c_j)$ is $O(1/\eps)$.  Thus
each $\frac{\partial}{\partial \eps}g(x_{-m}^{-1}, c_{-m}^{-1})$
is  $mO(1/\eps)$.  Each  $g(x_{-m}^{-1}, c_{-m}^{-1})$ is lower bounded by a positive constant,
uniformly over all $p \in \M_{\delta_0}$.  Thus, each $\frac{\partial}{\partial \eps}  g(x_{-m}^{-1}, c_{-m}^{-1})/ g(x_{-m}^{-1}, c_{-m}^{-1})$ is $m O(1/\eps)$. It then follows from (\ref{sum}) that $\frac{\partial}{\partial q} p(z_{-m}^i)/p(z_{-m}^i) = mO(1/\eps)$, as desired. For the partial derivatives with respect to $\vec{p}$, we observe that $\frac{\partial}{\partial q}p(x_{-m})/p(x_{-m})$ and $\frac{\partial}{\partial q} p(x_{j+1}|x_j)/p(x_{j+1}|x_j)$ (here, $q$ is a component of $\vec{p}$) are $O(1)$, with uniform constant over all $p \in \M_{\delta_0}$. We then immediately establish (\ref{prelim}) for $|\vec{k}|=1$.

We now prove (\ref{prelim}) for a generic $\vec{k}$.

Apply the multivariate Faa Di Bruno formula (for the derivatives of a composite function)~\cite{co96, le96} to the function $f(y)=1/y$ (here, $y$ is a function), we have for $\vec{l}$ with $|\vec{l}| \neq 0$,
$$
f(y)^{(\vec{l})}=\sum
D(\vec{a}_1, \vec{a}_2, \cdots, \vec{a}_t) (1/y) (y^{(\vec{a}_1)}/y) (y^{(\vec{a}_2)}/y) \cdots (y^{(\vec{a}_t)}/y),
$$
where the summation is over the set of unordered sequences of non-negative vectors $\vec{a}_1, \vec{a}_2,
\cdots, \vec{a}_t$ with $\vec{a}_1+\vec{a}_2+\cdots+\vec{a}_t=\vec{l}$ and $D(\vec{a}_1, \vec{a}_2, \cdots, \vec{a}_t)$ is the corresponding coefficient. For any $\vec{l}$, define $\vec{l}!=\prod_{i=1}^{|\mathcal{S}_2|+1} l_i!$; and for any $\vec{l} \preceq \vec{k}$ (every component of $\vec{l}$ is less or equal to the corresponding one of $\vec{k}$), define
$C_{\vec{k}}^{\vec{l}}=\vec{k}!/(\vec{l}!(\vec{k}-\vec{l})!)$. Then for any $\vec{k}$, applying the multivariate Leibnitz rule, we have
$$
\left(\frac{p((x_i, c_i), z_{-m}^i)}{p(z_{-m}^i)}\right)^{(\vec{k})}=\sum_{\vec{l} \preceq \vec{k}} C_{\vec{k}}^{\vec{l}} (p((x_i, c_i), z_{-m}^i))^{(\vec{k}-\vec{l})} (1/p(z_{-m}^i))^{(\vec{l})}
$$
$$
\hspace{-1cm}
=\sum_{\vec{l} \preceq \vec{k}} \sum_{\vec{a}_1+\vec{a}_2+\cdots+\vec{a}_t=\vec{l}} C_{\vec{k}}^{\vec{l}} D(\vec{a}_1, \cdots, \vec{a}_t) \frac{p((x_i, c_i), z_{-m}^i)}{p(z_{-m}^i)} \frac{p((x_i, c_i), z_{-m}^i)^{(\vec{k}-\vec{l})}}{p((x_i, c_i), z_{-m}^i)} \frac{p(z_{-m}^i)^{(\vec{a}_1)}}{p(z_{-m}^i)} \cdots \frac{p(z_{-m}^i)^{(\vec{a}_t)}}{p(z_{-m}^i)}.
$$
Then, similarly as above, one can show that
\begin{equation}  \label{wordsINequation}
p(z_{-m}^i)^{(\vec{a})}/p(z_{-m}^i), \qquad p((x_i, c_i), z_{-m}^i)^{(\vec{a})}/p((x_i, c_i), z_{-m}^i) = m^{|\vec{a}|} O(1/\eps^{|\vec{a}|}),
\end{equation}
which implies that there is a positive constant $C_{|\vec{k}|}$ such that
$$
|w_{i, -m}^{(\vec{k})}| \leq n^{|\vec{k}|}C_{|\vec{k}|}/\eps^{|\vec{k}|}.
$$
Obviously, the same argument can be applied to upper bound $|\hat{w}_{i, -\hat{m}}^{(\vec{k})}|$.

We next prove that, for each $\vec{k}$,
\begin{equation} \label{xiang-jian}
|w^{(\vec{k})}_{-1, -m}-\hat{w}^{(\vec{k})}_{-1, -\hat{m}}| = \hat{O}(\eps^n) \mbox{ on } \mathcal{M}_{\delta_0} \times T_{n, m, \hat{m}}^{\alpha}.
\end{equation}
Proposition~\ref{overwhelm} will then follow from (\ref{negative1}).

We first prove this for $|\vec{k}| =1$. Again, let $q$ be a component of $\vec{q}=(\vec{p}, \eps)$. Then, for $i=-1, -2, \cdots, -n_0$, we have
\begin{equation}  \label{k1-1}
\frac{\partial}{\partial q} w_{(i+1) N-1, -m} = \frac{\partial f_{[z]_i}}{\partial
w}(\vec{q}, w_{i N-1, -m}) \frac{\partial}{\partial q}w_{i N-1, -m}+\frac{\partial f_{[z]_i}}{\partial
q}(\vec{q}, w_{i N-1, -m}),
\end{equation}
and
\begin{equation}  \label{k1-2}
\frac{\partial}{\partial q}\hat{w}_{(i+1)N-1, -\hat{m}} = \frac{\partial f_{[z]_i}}{\partial
w}(\vec{q}, \hat{w}_{i N-1, -\hat{m}}) \frac{\partial}{\partial q}\hat{w}_{i N-1, -\hat{m}}+\frac{\partial f_{[z]_i}}{\partial q}(\vec{q}, \hat{w}_{i N-1, -\hat{m}}).
\end{equation}
Taking the difference, we then have
$$
\hspace{-0.5cm} \frac{\partial}{\partial q}w_{(i+1)N-1, -m} - \frac{\partial}{\partial q}\hat{w}_{(i+1)N-1,
-\hat{m}}=\frac{\partial f_{[z]_i}}{\partial q}(\vec{q},
w_{i N-1, -m})-\frac{\partial f_{[z]_i}}{\partial q}(\vec{q}, \hat{w}_{i N-1, -\hat{m}})
$$
$$
+\frac{\partial f_{[z]_i}}{\partial w}(\vec{q}, w_{i N-1, -m})
\frac{\partial}{\partial q}w_{i N-1, -m} - \frac{\partial f_{[z]_i}}{\partial w}(\vec{q}, \hat{w}_{i N-1, -\hat{m}}) \frac{\partial}{\partial q}\hat{w}_{i N-1, -\hat{m}}
$$
$$
= \left(\frac{\partial f_{[z]_i}}{\partial
q}(\vec{q}, w_{i N-1, -m})-\frac{\partial f_{[z]_i}}{\partial
q}(\vec{q}, \hat{w}_{i N-1, -\hat{m}}) \right)
$$
$$
+ \left(\frac{\partial f_{[z]_i}}{\partial
w}(\vec{q}, w_{i N-1, -m}) \frac{\partial}{\partial q}w_{i N-1, -m}
-\frac{\partial f_{[z]_i}}{\partial
w}(q, \hat{w}_{i N-1, -\hat{m}}) \frac{\partial}{\partial q}w_{i N-1, -m}  \right)
$$
$$
+ \left(\frac{\partial f_{[z]_i}}{\partial
w}(\vec{q}, \hat{w}_{i N-1, -\hat{m}}) \frac{\partial}{\partial q}w_{i N-1, -m} - \frac{\partial f_{[z]_i}}{\partial
w}(\vec{q}, \hat{w}_{i N-1, -\hat{m}}) \frac{\partial}{\partial q}\hat{w}_{i N-1, -\hat{m}} \right).
$$
This last expression is the sum of three terms, which we will refer to as $T_1$, $T_2$ and $T_3$.

From Lemma~\ref{Lipschitz}, one
checks that for all $[z]_i \in \mathcal{Z}^N$, $w \in \mathcal{W}$ and $\vec{q} \in U_{\delta_0, \eps_0}$,
$$
\left| \frac{\partial^2 f_{[z]_i}}{\partial \vec{q} \partial w}(\vec{q}, w) \right|, \left|
\frac{\partial^2 f_{[z]_i}}{\partial w \partial w}(\vec{q}, w) \right| \leq C/\eps^{4 N O_M}.
$$
(Here, we remark that there are many different constants in this proof, which we
will often refer to using the same notation $C$, making sure that the dependence of these constants on various parameters
is clear.) It then follows from the mean value theorem  that for each $i=-1, -2, \cdots, -n_0$
$$
T_1 \le  (C/\eps^{4 N O_M}) |w_{i N-1, -m}- \hat{w}_{i N-1, -\hat{m}}|.
$$
By the mean value theorem and (\ref{prelim}),
$$
T_2 \le (C/\eps^{4 N O_M}) (nC_1/\eps) |w_{i N-1, -m}- \hat{w}_{i N-1, -\hat{m}}|.
$$
And finally
$$
T_3 \le \left\| \frac{\partial f_{[z]_i}}{\partial
w}(\vec{q}, \hat{w}_{i N-1, -\hat{m}}) \right\| |\frac{\partial}{\partial q}w_{i N-1, -m} - \frac{\partial}{\partial q}\hat{w}_{i N-1, -\hat{m}}|.
$$
Thus,
$$
|\frac{\partial}{\partial q}w_{(i+1)N-1, -m}-\frac{\partial}{\partial q}\hat{w}_{(i+1)N-1, -\hat{m}}| \leq \left\| \frac{\partial f_{[z]_i}}{\partial
w}(\vec{q}, \hat{w}_{i N-1, -\hat{m}}) \right\| |\frac{\partial}{\partial q}w_{i N-1, -m}- \frac{\partial}{\partial q}\hat{w}_{i N-1, -\hat{m}}|
$$
$$
+(1+nC_1/\eps)C\eps^{-4 N O_M} |w_{i N-1, -m}-\hat{w}_{i N-1, -\hat{m}}|.
$$
Iteratively apply this inequality to obtain
$$
\hspace{-2cm} |\frac{\partial}{\partial q}w_{-1, -m}-\frac{\partial}{\partial q}\hat{w}_{-1, -\hat{m}}| \leq
\prod_{i=-n_0}^{-1} \left\|\frac{\partial f_{[z]_{i}}}{\partial w}(\vec{q}, \hat{w}_{i N-1, -\hat{m}}) \right\| |\frac{\partial}{\partial q}w_{-n_0 N-1, -m}-
\frac{\partial}{\partial q}\hat{w}_{-n_0 N -1, -\hat{m}}|
$$
$$
+\prod_{i=-n_0+1}^{-1}  \left\|\frac{\partial f_{[z]_i}}{\partial w}(\vec{q}, \hat{w}_{i N-1, -\hat{m}}) \right\| (1+nC_1/\eps)C\eps^{-4 N O_M} |w_{-n_0 N -1, -m}- \hat{w}_{-n_0 N -1, -\hat{m}}|
$$
$$
+\cdots+\prod_{i=-j}^{-1}  \left\|\frac{\partial f_{[z]_i}}{\partial w}(\vec{q}, \hat{w}_{i N-1, -\hat{m}}) \right\| (1+nC_1/\eps)C\eps^{-4 N O_M} |w_{(-j-1) N -1, -m}- \hat{w}_{(-j-1)N -1, -\hat{m}}|+
$$
$$
+\cdots+\left\|\frac{\partial f_{[z]_{-1}}}{\partial w}(\vec{q}, \hat{w}_{-N-1, -\hat{m}}) \right\| (1+nC_1/\eps)C\eps^{-4 N O_M} |w_{-2N-1 , -m}- \hat{w}_{-2N-1, -\hat{m}}|
$$
\begin{equation} \label{haofan}
+(1+nC_1/\eps)C\eps^{-4 N O_M} |w_{-N-1, -m}- \hat{w}_{-N-1, -\hat{m}}|.
\end{equation}

Now, apply the mean value theorem, we deduce that there exist $\xi_i$, $-n_0 \leq i \leq -j-2$ (here $\xi_i$ is a convex combination of $w_{-i N-1, -m}$ and $\hat{w}_{-i N -1, -\hat{m}}$) such that
$$
|w_{(-j-1)N-1, -m}-\hat{w}_{(-j-1)N-1, -\hat{m}}|=|f_{[z]_{-n_0}^{-j-2}}(w_{-n_0 N-1, -m})-f_{[z]_{-n_0}^{-j-2}}(\hat{w}_{-n_0 N-1, -\hat{m}})|
$$
$$
\leq \prod_{i=-n_0}^{-j-2} \|D_w f_{[z]_i}(\xi_i)\| |w_{-n_0 N-1, -m}-\hat{w}_{-n_0 N-1, -\hat{m}}|.
$$
Then, recall that an $\alpha$-typical sequence $z_{-n}^{-1}$ breaks at most $3 \alpha n$ times.
Thus there are at least $(1-3\alpha)n$ $i$'s where we can use the estimate (\ref{continue}) and
at most $3\alpha n$ $i$'s where we can only use the weaker estimates (\ref{break}). Similar to the derivation of (\ref{CcCe}), with Remark~\ref{TheSameEps}, we derive that for any $\alpha < \alpha_0$, every term in the right hand side of (\ref{haofan}) is $\hat{O}(\eps^n)$ on $\mathcal{M}_{\delta_0} \times T_{n, m, \hat{m}}^{\alpha}$ (we use (\ref{prelim}) to upper bound the first term). Again, with Remark~\ref{TheSameEps}, we conclude that
$$
\left|\frac{\partial w_{-1, -m}}{\partial \vec{q}}-\frac{\partial \hat{w}_{-1, -\hat{m}}}{\partial \vec{q}}\right|= \hat{O}(\eps^n) \mbox{ on } \mathcal{M}_{\delta_0} \times T_{n, m, \hat{m}}^{\alpha},
$$
which, by (\ref{negative1}), implies the proposition for $\ell=1$, as desired.

The proof of (\ref{xiang-jian}) for a generic $\vec{k}$ is rather similar, however very tedious. We next briefly illustrate the idea of the proof. Note that (compare with (\ref{k1-1}), (\ref{k1-2}) for $|\vec{k}|=1$)
$$
w_{(i+1) N-1, -m}^{(\vec{k})} = \frac{\partial f_{[z]_i}}{\partial
w}(\vec{q}, w_{i N-1, -m}) w_{i N-1, -m}^{(\vec{k})}+\mbox{ others}
$$
and
$$
\hat{w}_{(i+1)N-1, -\hat{m}}^{(\vec{k})} = \frac{\partial f_{[z]_i}}{\partial
w}(\vec{q}, \hat{w}_{i N-1, -\hat{m}}) \hat{w}_{i N-1, -\hat{m}}^{(\vec{k})} +\mbox{ others},
$$
where the first ``others'' is a linear combination of terms taking the following forms (below, $t$ can be $0$, which corresponds to the partial derivatives of $f$ with respect to the first argument $\vec{q}$):
$$
f_{[z]_i}^{(\vec{k}')}(\vec{q}, w_{i N-1, -m}) w_{i N-1, -m}^{(\vec{a}_1)} \cdots w_{i N-1, -m}^{(\vec{a}_t)},
$$
and the second ``others'' is a linear combination of terms taking the following forms:
$$
f_{[z]_i}^{(\vec{k}')}(\vec{q}, \hat{w}_{i N-1, -\hat{m}}) \hat{w}_{i N-1, -\hat{m}}^{(\vec{a}_1)} \cdots \hat{w}_{i N-1, -\hat{m}}^{(\vec{a}_t)},
$$
here $\vec{k}' \preceq \vec{k}$, $t \leq |\vec{k}|$ and $|\vec{a}_i| < |\vec{k}|$ for all $i$. Using (\ref{prelim}) and the fact that there exists a constant $C$ (by Lemma~\ref{Lipschitz}) such that
$$
|f_{[z]_i}^{(\vec{k}')}(\vec{q}, w_{i N-1, -m})| \leq C/\eps^{4NO_M |\vec{k}'|},
$$
we then can establish (compare with (\ref{haofan}) for $|\vec{k}|=1$)
$$
|w_{(i+1)N-1, -m}^{(k)}-\hat{w}_{(i+1)N-1, -\hat{m}}^{(k)}| \leq \left\| \frac{\partial f_{[z]_i}}{\partial
w}(\vec{q}, \hat{w}_{i N-1, -\hat{m}}) \right\| |w_{i N-1, -m}^{\vec{k}}-\hat{w}_{i N-1, -\hat{m}}^{(\vec{k})}|+ \mbox{ others},
$$
where ``others'' is the sum of finitely many terms, each of which takes the following form (see the $j$-th term of (\ref{haofan}) for $|\vec{k}|=1$)
\begin{equation} \label{rough-form}
n^{D_{\vec{k}'}} O(1/\eps^{D_{\vec{k}'}}) \prod_{i=-j}^{-1}  \left\|\frac{\partial f_{[z]_i}}{\partial w}(\vec{q}, \hat{w}_{i N-1, -\hat{m}}) \right\| |w_{(-j-1) N-1, -m}^{(\vec{a})}-\hat{w}_{(-j-1) N-1, -\hat{m}}^{(\vec{a})}|,
\end{equation}
where $|\vec{a}| < |\vec{k}|$, $D_{\vec{k}'}$ is a constant dependent on $\vec{k}'$. Then inductively, one can
use the similar approach to establish that (\ref{rough-form}) is $\hat{O}(\eps^n)$ on $\mathcal{M}_{\delta_0} \times T_{n, m, \hat{m}}^{\alpha}$, which implies (\ref{xiang-jian}) for a generic $\vec{k}$, and thus the proposition for a generic $\ell$.

\end{proof}

\subsection{Asymptotic behavior of entropy rate}

The parameterization of $Z$ as a function of $\eps$ fits in the framework of~\cite{hm08} in a more general setting. Consequently, we have the following three propositions.

\begin{pr}  \label{conditional}
Assume that $\vec{p} \in \mathcal{M}_0$. For any sequence $z_{-n}^0 \in \mathcal{Z}^{n+1}$, $p((x_{-1}, c_{-1})=\cdot|z_{-n}^{-1})$ and $p(z_0|z_{-n}^{-1})$ are analytic around $\eps=0$. Moreover, $\ord(p(z_0|z_{-n}^{-1})) \leq O_M$.
\end{pr}

\begin{proof}
Analyticity of $p((x_{-1}, c_{-1})=\cdot|z_{-n}^{-1})$ follows from Proposition 2.4 in~\cite{hm08}. It then follows from $p(z_0|z_{-n}^{-1})=p((x_{-1}, c_{-1})=\cdot|z_{-n}^{-1}) \Omega_{z_{0}} \mathbf{1}$ and the fact that any row sum of $\Omega_{z_{0}}$ is non-zero that $p(z_0|z_{-n}^{-1})$ is analytic with $\ord(p(z_0|z_{-n}^{-1})) \leq O_M$.
\end{proof}

\begin{pr} (see Proposition 2.7 in~\cite{hm08})   \label{StabilizingLemma}
Assume that $\vec{p} \in \mathcal{M}_0$. For two fixed hidden Markov chain sequences $z_{-m}^0, \hat{z}_{-\hat{m}}^0$ such that
$$
z_{-n}^0=\hat{z}_{-n}^0, \qquad \ord(p(z_{-n}^{-1}|z_{-m}^{-n-1})), \;\; \ord(p(\hat{z}_{-n}^{-1}|\hat{z}_{-\hat{m}}^{-n-1})) \leq k
$$
for some $n \leq m, \hat{m}$ and some $k$, we have for $j$ with
$0 \leq j \leq n-4k-1$,
$$
p^{(j)}(z_0|z_{-m}^{-1})(0)=p^{(j)}(\hat{z}_0|\hat{z}_{-\hat{m}}^{-1})(0),
$$
where the derivatives are taken with respect to $\eps$.
\end{pr}

\begin{rem} \label{StabilizingRemark}
It follows from Proposition~\ref{StabilizingLemma} that for any $\alpha$-typical sequence $z_{-n}^{-1}$ with $\alpha$ small enough and $n$ large enough, $\ord(p(z_0|z_{-n}^{-1}))=\ord(p(z_0|z_{-n-1}^{-1}))$
\end{rem}

\begin{pr}  (see Theorem 2.8 in~\cite{hm08}) \label{main}
Assume that $\vec{p} \in \mathcal{M}_0$. For any $k \geq 0$,
\begin{equation}   \label{mainFormula}
H(Z)=H(Z)|_{\eps=0}+\sum_{j=1}^{k} g_j \eps^j + \sum_{j=1}^{k+1} f_j
\eps^j \log \eps + O(\eps^{k+1}),
\end{equation}
where $f_j$'s and $g_j$'s depend on $\Pi$ and $q_c$ (but not on $\eps$), the transition probability matrix of $X$.
\end{pr}

The following theorem strengthens Proposition~\ref{main} in the sense that it describes how the coefficients $f_j$'s and $g_j$'s vary with respect to the input Markov chain. We first introduce some necessary notation. We shall break $H_n(Z)$ into a sum of $G_n(Z)$ and $F_n(Z) \log(\eps)$ where $G_n(Z)=G_n(\vec{p}, \eps)$ and $F_n(Z)=F_n(\vec{p}, \eps)$ are smooth; precisely, we have
$$
H_n(Z)=G_n(\vec{p}, \eps)+F_n(\vec{p}, \eps) \log \eps,
$$
where ($\ord(p(z_0|z_{-n}^{-1}))$ is well-defined since $p(z_0|z_{-n}^{-1})$ is analytic with respect to $\eps$; see Proposition~\ref{conditional})
\begin{equation}
\label{Fn}
F_n(\vec{p}, \eps)=\sum_{z_{-n}^{0}} -\ord(p(z_0|z_{-n}^{-1}))
p(z_{-n}^0)
\end{equation}
and
\begin{equation}
\label{Gn}
G_n(\vec{p}, \eps)=\sum_{z_{-n}^{0}} -p(z_{-n}^0) \log
p^{\circ}(z_0|z_{-n}^{-1}),
\end{equation}
where
$$
p^{\circ}(z_0|z_{-n}^{-1}) = p(z_0|z_{-n}^{-1})/\eps^{\ord(p(z_0|z_{-n}^{-1}))}.
$$

\begin{thm} \label{main-1}
Given $\delta_0 >0$, for sufficiently small $\eps_0$,
\begin{enumerate}
\item On $U_{\delta_0, \eps_0}$, there is an analytic function $F(\vec{p}, \eps)$ and smooth (i.e., infinitely differentiable) function $G(\vec{p}, \eps)$ such that
\begin{equation}   \label{main-1Formula}
H(Z(\vec{p}, \eps))=G(\vec{p}, \eps)+F(\vec{p}, \eps) \log \eps.
\end{equation}
Moreover,
$$
G(\vec{p}, \eps)=H(Z)|_{\eps=0}+\sum_{j=1}^{k} g_j(\vec{p}) \eps^j+O(\eps^{k+1}), \qquad F(\vec{p}, \eps)=\sum_{j=1}^{k} f_j(\vec{p}) \eps^j+O(\eps^{k+1}),
$$
here $f_j$'s and $g_j$'s are the corresponding functions as in Proposition~\ref{main};
\item
Define $\hat{F}(\vec{p}, \eps) = F(\vec{p}, \eps)/\eps$. Then $\hat{F}(\vec{p}, \eps)$ is analytic on $U_{\delta_0, \eps_0}$.
\item
For any $\ell$, there exists $0 < \rho <1$ such that on $U_{\delta_0, \eps_0}$
$$
|D_{\vec{p},\eps}^{\ell} F_n(\vec{p}, \eps) - D_{\vec{p},\eps}^{\ell} F(\vec{p},
\eps)| < \rho^n,
$$
$$
|D_{\vec{p},\eps}^{\ell} \hat{F}_n(\vec{p}, \eps) - D_{\vec{p},\eps}^{\ell} \hat{F}(\vec{p},
\eps)| < \rho^n,
$$
and
$$
|D_{\vec{p},\eps}^{\ell} G_n(\vec{p}, \eps) - D_{\vec{p},\eps}^{\ell} G(\vec{p},
\eps)| < \rho^n,
$$
for sufficiently large $n$.
\end{enumerate}
\end{thm}

\begin{proof}
\textbf{1)}
Recall that
$$
H_n(Z)=\sum_{z_{-n}^0} -p(z_{-n}^0) \log p(z_0|z_{-n}^{-1}).
$$
It follows from a compactness argument that $H_n(Z)$ uniformly converges to $H(Z)$ on the parameter space $U_{\delta_0,\eps_0}$ for any positive $\eps_0$. We now define
$$
H_n^{\alpha}(Z)=\sum_{z_{-n}^{-1} \in T_{n}^{\alpha}, z_0} -p(z_{-n}^0) \log p(z_0|z_{-n}^{-1});
$$
here recall that $T_{n}^{\alpha}$ denotes the set of all $\alpha$-typical $\mathcal{Z}$-sequences with length $n$.
Applying Lemma~\ref{MATH2603}, we deduce that $H_n^{\alpha}(Z)$ uniformly converges to $H(Z)$ on $U_{\delta_0,\eps_0}$ as well.

By Proposition~\ref{conditional}, $p(z_0|z_{-n}^{-1})$ is analytic with $\ord(p(z_0|z_{-n}^{-1})) \leq O_M$. It then follows that for any $\alpha$ with $0 < \alpha < 1$ (we will choose $\alpha$ to be smaller later if necessary),
$$
H_n^{\alpha}(Z)=G_n^{\alpha}(\vec{p}, \eps)+F_n^{\alpha}(\vec{p}, \eps) \log \eps,
$$
where
$$
F_n^{\alpha}(\vec{p}, \eps)=\sum_{z_{-n}^{-1} \in T_{n}^{\alpha}, z_0} -\ord(p(z_0|z_{-n}^{-1})) p(z_{-n}^0),
$$
and
$$
G_n^{\alpha}(\vec{p}, \eps)=\sum_{z_{-n}^{-1} \in T_{n}^{\alpha}, z_0} -p(z_{-n}^0) \log p^{\circ}(z_0|z_{-n}^{-1}).
$$

The idea of the proof is as follows. We first show that $F_n^{\alpha}(\vec{p}, \eps)$ uniformly converges to a real analytic function $F(\vec{p}, \eps)$. We then prove that $G_n^{\alpha}(\vec{p}, \eps)$ and its derivatives with respect to $(\vec{p}, \eps)$ also uniformly converge to a smooth function $G(\vec{p}, \eps)$. Since $H_n^{\alpha}(Z)$ uniformly converges to $H(Z)$, $F(\vec{p}, \eps)$, $G(\vec{p}, \eps)$ satisfy (\ref{main-1Formula}). The ``Moreover'' part then immediately follows by equating (\ref{mainFormula}) and (\ref{main-1Formula}) to compare the coefficients.

We now show that $F_n^{\alpha}(\vec{p}, \eps)$ uniformly converges to a real analytic function $F(\vec{p}, \eps)$. Now
$$
|F_n^{\alpha}(\vec{p}, \eps)-F_{n+1}^{\alpha}(\vec{p}, \eps)|=\left| \sum_{z_{-n}^{-1} \in T_{n}^{\alpha}, z_0} \ord(p(z_0|z_{-n}^{-1})) p(z_{-n}^0)-\sum_{z_{-n-1}^{-1} \in T_{n+1}^{\alpha}, z_0} \ord(p(z_0|z_{-n-1}^{-1})) p(z_{-n-1}^0) \right|
$$
$$
=\left| \left(\sum_{z_{-n}^{-1} \in T_{n}^{\alpha}, z_{-n-1}^{-1} \in T_{n+1}^{\alpha}, z_0}+\sum_{z_{-n}^{-1} \in T_{n}^{\alpha}, z_{-n-1}^{-1} \not \in T_{n+1}^{\alpha}, z_0} \right) \ord(p(z_0|z_{-n}^{-1})) p(z_{-n-1}^0) \right.
$$
$$
-\left. \left(\sum_{z_{-n}^{-1} \in T_{n}^{\alpha}, z_{-n-1}^{-1} \in T_{n+1}^{\alpha}, z_0}+\sum_{z_{-n}^{-1} \not \in T_{n}^{\alpha}, z_{-n-1}^{-1} \in T_{n+1}^{\alpha}, z_0} \right) \ord(p(z_0|z_{-n-1}^{-1})) p(z_{-n-1}^0) \right|.
$$
By Remark~\ref{StabilizingRemark}, we have
$$
|F_n^{\alpha}(\vec{p}, \eps)-F_{n+1}^{\alpha}(\vec{p}, \eps)|=\left| \sum_{z_{-n}^{-1} \in T_{n}^{\alpha}, z_{-n-1}^{-1} \not \in T_{n+1}^{\alpha}, z_0} \ord(p(z_0|z_{-n}^{-1})) p(z_{-n-1}^0) \right.
$$
$$
-\left. \sum_{z_{-n}^{-1} \not \in T_{n}^{\alpha}, z_{-n-1}^{-1} \in T_{n+1}^{\alpha}, z_0} \ord(p(z_0|z_{-n-1}^{-1})) p(z_{-n-1}^0) \right|.
$$
Applying Lemma~\ref{MATH2603}, we have
\begin{equation} \label{F}
|F_n^{\alpha}(\vec{p}, \eps)-F_{n+1}^{\alpha}(\vec{p}, \eps)| = \hat{O}(\eps^n) \mbox{ on } \mathcal{M}_{\delta_0},
\end{equation}
which implies that there exists $\eps_0 > 0$ such that $F_n^{\alpha}(\vec{p}, \eps)$ are exponentially Cauchy and thus uniformly converges on $U_{\delta_0, \eps_0}$ to a continuous function $F(\vec{p}, \eps)$.

Let $F_n^{\alpha, \mathbb{C}}(\vec{p}, \eps)$ denote the complexified $F_n^{\alpha}(\vec{p}, \eps)$ on $(\vec{p}, \eps)$ with $\vec{p} \in \mathcal{M}_{\delta_0}^{\mathbb{C}}(\eta_0)$ and $|\eps| \leq \eps_0$. Then, using Lemma~\ref{complex} and a similar argument as above, we can prove that
\begin{equation} \label{FC}
|F_n^{\alpha, \mathbb{C}}(\vec{p}, \eps)-F_{n+1}^{\alpha, \mathbb{C}}(\vec{p}, \eps)| = \hat{O}(|\eps|^n) \mbox{ on } \mathcal{M}_{\delta_0}^{\mathbb{C}}(\eta_0);
\end{equation}
in other words, for some $\eta_0, \eps_0 > 0$, $F_n^{\alpha, \mathbb{C}}(\vec{p}, \eps)$ are exponentially Cauchy and thus uniformly converges on all $(\vec{p}, \eps)$ with $\vec{p} \in \mathcal{M}_{\delta_0}^{\mathbb{C}}(\eta_0)$ and $|\eps| \leq \eps_0$. Therefore, $F(\vec{p}, \eps)$ is analytic with respect to $(\vec{p}, \eps)$ on $U_{\delta_0, \eps_0}$.

We now prove that $G_n^{\alpha}(\vec{p}, \eps)$ and its derivatives with respect to $(\vec{p}, \eps)$ uniformly converge to a smooth function $G^{\alpha}(\vec{p}, \eps)$ and its derivatives.

Although the convergence of $G_n^{\alpha}(\vec{p}, \eps)$ and its derivatives can be proven through the same argument at once, we first prove the convergence of $G_n^{\alpha}(\vec{p}, \eps)$ only for illustrative purpose.

For any $\alpha, \beta > 0$, we have
\begin{equation} \label{ine-1}
|\log \alpha - \log \beta| \leq \max\{|(\alpha-\beta)/\beta|, |(\alpha-\beta)/\alpha|\}.
\end{equation}
Note that the following is contained in Proposition~\ref{overwhelm} ($\ell=0$)
\begin{equation} \label{ine-2}
|p^{\circ}(z_0|z_{-n}^{-1})-p^{\circ}(z_0|z_{-n-1}^{-1})| = \hat{O}(\eps^n) \mbox{ on } \mathcal{M}_{\delta_0} \times T_{n, n, n+1}^{\alpha}.
\end{equation}
One further checks that by Proposition~\ref{conditional}, there exists a positive constant $C$ such that for $\eps$ small enough and for any sequence $z_{-n}^{-1}$,
$$
p(z_0|z_{-n}^{-1}) \geq C\eps^{O_M},
$$
and thus,
\begin{equation}  \label{Austin-3}
p^{\circ}(z_0|z_{-n}^{-1}) \geq C \eps^{O_M}.
\end{equation}

Using (\ref{ine-1}), (\ref{ine-2}), (\ref{Austin-3}) and Lemma~\ref{MATH2603}, we have
$$
|G_n^{\alpha}(\vec{p}, \eps)-G_{n+1}^{\alpha}(\vec{p}, \eps)|=\left| \sum_{z_{-n}^{-1} \in T_{n}^{\alpha}, z_0} -p(z_{-n}^0) \log p^{\circ}(z_0|z_{-n}^{-1})-\sum_{z_{-n-1}^{-1} \in T_{n+1}^{\alpha}, z_0} -p(z_{-n-1}^0) \log p^{\circ}(z_0|z_{-n-1}^{-1}) \right|
$$
$$
=\left| \left(\sum_{z_{-n}^{-1} \in T_{n}^{\alpha}, z_{-n-1}^{-1} \in T_{n+1}^{\alpha}, z_0}+\sum_{z_{-n}^{-1} \in T_{n}^{\alpha}, z_{-n-1}^{-1} \not \in T_{n+1}^{\alpha}, z_0} \right) -p(z_{-n-1}^0) \log p^{\circ}(z_0|z_{-n}^{-1}) \right.
$$
$$
-\left. \left(\sum_{z_{-n}^{-1} \in T_{n}^{\alpha}, z_{-n-1}^{-1} \in T_{n+1}^{\alpha}, z_0}+\sum_{z_{-n}^{-1} \not \in T_{n}^{\alpha}, z_{-n-1}^{-1} \in T_{n+1}^{\alpha}, z_0} \right) -p(z_{-n-1}^0) \log p^{\circ}(z_0|z_{-n-1}^{-1}) \right|
$$
$$
\leq \left| \sum_{z_{-n}^{-1} \in T_{n}^{\alpha}, z_{-n-1}^{-1} \in T_{n+1}^{\alpha}, z_0} -p(z_{-n-1}^0) (\log p^{\circ}(z_0|z_{-n}^{-1})-\log p^{\circ}(z_0|z_{-n-1}^{-1})) \right|
$$
$$
+ \left| \sum_{z_{-n}^{-1} \in T_{n}^{\alpha}, z_{-n-1}^{-1} \not \in T_{n+1}^{\alpha}, z_0} -p(z_{-n-1}^0) \log p^{\circ}(z_0|z_{-n}^{-1}) \right| + \left| \sum_{z_{-n}^{-1} \not \in T_{n}^{\alpha}, z_{-n-1}^{-1} \in T_{n+1}^{\alpha}, z_0} -p(z_{-n-1}^0) \log p^{\circ}(z_0|z_{-n-1}^{-1}) \right|
$$
$$
\hspace{-1cm} \leq \sum_{z_{-n}^{-1} \in T_{n}^{\alpha}, z_{-n-1}^{-1} \in T_{n+1}^{\alpha}, z_0} p(z_{-n-1}^0) \max \left\{\left| \frac{ p^{\circ}(z_0|z_{-n}^{-1})-p^{\circ}(z_0|z_{-n-1}^{-1})}{p^{\circ}(z_0|z_{-n-1}^{-1})} \right|, \left| \frac{ p^{\circ}(z_0|z_{-n}^{-1})-p^{\circ}(z_0|z_{-n-1}^{-1})}{p^{\circ}(z_0|z_{-n}^{-1})} \right| \right\}
$$
\begin{equation}  \label{for-zero}
\hspace{-2cm} + \left| \sum_{z_{-n}^{-1} \in T_{n}^{\alpha}, z_{-n-1}^{-1} \not \in T_{n+1}^{\alpha}, z_0} -p(z_{-n-1}^0) \log p^{\circ}(z_0|z_{-n}^{-1}) \right| + \left| \sum_{z_{-n}^{-1} \not \in T_{n}^{\alpha}, z_{-n-1}^{-1} \in T_{n+1}^{\alpha}, z_0} -p(z_{-n-1}^0) \log p^{\circ}(z_0|z_{-n-1}^{-1}) \right| = \hat{O}(\eps^n) \mbox{ on } \mathcal{M}_{\delta_0},
\end{equation}
which implies that there exists $\eps_0 > 0$ such that $G_n^{\alpha}(\vec{p}, \eps)$ uniformly converges on $U_{\delta_0, \eps_0}$, then the existence of $G(\vec{p}, \eps)$ immediately follows.

Apply the multivariate Faa Di Bruno formula~\cite{co96, le96} to the function $f(y)=\log y$, we have for $\vec{l}$ with $|\vec{l}| \neq 0$,
$$
f(y)^{(\vec{l})}=\sum
D(\vec{a}_1, \vec{a}_2, \cdots, \vec{a}_k) (y^{(\vec{a}_1)}/y) (y^{(\vec{a}_2)}/y) \cdots (y^{(\vec{a}_k)}/y),
$$
where the summation is over the set of unordered sequences of non-negative vectors $\vec{a}_1, \vec{a}_2,
\cdots, \vec{a}_k$ with $\vec{a}_1+\vec{a}_2+\cdots+\vec{a}_k=\vec{l}$ and $D(\vec{a}_1, \vec{a}_2, \cdots, \vec{a}_k)$ is the corresponding coefficient. Then for any $\vec{m}$, applying the multivariate Leibnitz rule, we have
$$
(G_n^{\alpha})^{(\vec{m})}(\vec{p}, \eps)=\sum_{z_{-n}^{-1} \in T_{n}^{\alpha}, z_0} \sum_{\vec{l} \preceq \vec{m}}-C_{\vec{m}}^{\vec{l}} p^{(\vec{m}-\vec{l})}(z_{-n}^0) (\log p^{\circ}(z_0|z_{-n}^{-1}))^{(\vec{l})}
$$
$$
\hspace{-1cm}
=\sum_{z_{-n}^{-1} \in T_{n}^{\alpha}, z_0} \sum_{|\vec{l}| \neq 0, \vec{l} \preceq \vec{m}} \sum_{\vec{a}_1+\vec{a}_2+\cdots+\vec{a}_k=\vec{l}} -C_{\vec{m}}^{\vec{l}} D(\vec{a}_1, \cdots, \vec{a}_k) p^{(\vec{m}-\vec{l})}(z_{-n}^0) \frac{p^{\circ}(z_0|z_{-n}^{-1})^{(\vec{a}_1)}}{p^{\circ}(z_0|z_{-n}^{-1})} \cdots \frac{p^{\circ}(z_0|z_{-n}^{-1})^{(\vec{a}_k)}}{p^{\circ}(z_0|z_{-n}^{-1})}
$$
\begin{equation}  \label{LastFirst}
+\sum_{z_{-n}^{-1} \in T_{n}^{\alpha}, z_0} -p^{(\vec{m})}(z_{-n}^0) \log p^{\circ}(z_0|z_{-n}^{-1}).
\end{equation}

We tackle the last term of (\ref{LastFirst}) first. Using (\ref{ine-1}) and (\ref{ine-2}) and with a parallel argument obtained through replacing $p(z_{-n}^0), p(z_{-n-1}^0)$ in (\ref{for-zero}) by $p^{(\vec{m})}(z_{-n}^0), p^{(\vec{m})}(z_{-n-1}^0)$, respectively, we can show that
$$
\hspace{-1cm} \left| \sum_{z_{-n}^{-1} \in T_{n}^{\alpha}, z_0} -p^{(\vec{m})}(z_{-n}^0) \log p^{\circ}(z_0|z_{-n}^{-1})-\sum_{z_{-n-1}^{-1} \in T_{n+1}^{\alpha}, z_0} -p^{(\vec{m})}(z_{-n-1}^0) \log p^{\circ}(z_0|z_{-n-1}^{-1}) \right| = \hat{O}(\eps^n) \mbox{ on } \mathcal{M}_{\delta_0} \times T_{n, n, n+1}^{\alpha},
$$
where we used the fact that for any $z_{-n}^0$ and $\vec{m}$, $p^{(\vec{m})}(z_{-n}^0)/p(z_{-n}^0)$ is $O(n^{|\vec{m}|}/ {\eps}^{|\vec{m}|})$ (see (\ref{wordsINequation})).
And using the identity
$$
\alpha_1 \alpha_2 \cdots \alpha_n - \beta_1 \beta_2 \cdots \beta_n=(\alpha_1-\beta_1) \alpha_2 \cdots \alpha_n+
\beta_1 (\alpha_2-\beta_2) \alpha_3 \cdots \alpha_n+\cdots+\beta_1 \cdots \beta_{n-1} (\alpha_n-\beta_n),
$$
we have
$$
\left| \frac{p^{\circ}(z_0|z_{-n}^{-1})^{(\vec{a}_1)}}{p^{\circ}(z_0|z_{-n}^{-1})} \cdots \frac{p^{\circ}(z_0|z_{-n}^{-1})^{(\vec{a}_k)}}{p^{\circ}(z_0|z_{-n}^{-1})}- \frac{p^{\circ}(z_0|z_{-n-1}^{-1})^{(\vec{a}_1)}}{p^{\circ}(z_0|z_{-n-1}^{-1})} \cdots \frac{p^{\circ}(z_0|z_{-n-1}^{-1})^{(\vec{a}_k)}}{p^{\circ}(z_0|z_{-n-1}^{-1})}\right|
$$
$$
\leq \left| \left(\frac{p^{\circ}(z_0|z_{-n}^{-1})^{(\vec{a}_1)}}{p^{\circ}(z_0|z_{-n}^{-1})}-\frac{p^{\circ}(z_0|z_{-n-1}^{-1})^{(\vec{a}_1)}}{p^{\circ}(z_0|z_{-n-1}^{-1}) }\right) \frac{p^{\circ}(z_0|z_{-n}^{-1})^{(\vec{a}_2)}}{p^{\circ}(z_0|z_{-n}^{-1})} \cdots \frac{p^{\circ}(z_0|z_{-n}^{-1})^{(\vec{a}_k)}}{p^{\circ}(z_0|z_{-n}^{-1})} \right|
$$
$$
+\left| \frac{p^{\circ}(z_0|z_{-n-1}^{-1})^{(\vec{a}_1)}}{p^{\circ}(z_0|z_{-n-1}^{-1})} \left(\frac{p^{\circ}(z_0|z_{-n}^{-1})^{(\vec{a}_2)}}{p^{\circ}(z_0|z_{-n}^{-1})}-\frac{p^{\circ}(z_0|z_{-n-1}^{-1})^{(\vec{a}_2)}}{p^{\circ}(z_0|z_{-n-1}^{-1})} \right) \frac{p^{\circ}(z_0|z_{-n}^{-1})^{(\vec{a}_3)}}{p^{\circ}(z_0|z_{-n}^{-1})} \cdots \frac{p^{\circ}(z_0|z_{-n}^{-1})^{(\vec{a}_k)}}{p^{\circ}(z_0|z_{-n}^{-1})} \right|+\cdots
$$
$$
+ \left| \frac{p^{\circ}(z_0|z_{-n-1}^{-1})^{(\vec{a}_1)}}{p^{\circ}(z_0|z_{-n-1}^{-1})} \cdots \frac{p^{\circ}(z_0|z_{-n-1}^{-1})^{(\vec{a}_{k-1})}}{p^{\circ}(z_0|z_{-n-1}^{-1})} \left(\frac{p^{\circ}(z_0|z_{-n}^{-1})^{(\vec{a}_k)}}{p^{\circ}(z_0|z_{-n}^{-1})}-\frac{p^{\circ}(z_0|z_{-n-1}^{-1})^{(\vec{a}_k)}}{p^{\circ}(z_0|z_{-n-1}^{-1})}\right) \right|.
$$

Now apply the inequality
$$
\left| \frac{\beta_1}{\alpha_1}-\frac{\beta_2}{\alpha_2} \right|= \left| \frac{\beta_1}{\alpha_1}-\frac{\beta_1}{\alpha_2}+\frac{\beta_1}{\alpha_2}-\frac{\beta_2}{\alpha_2} \right| \leq |\beta_1/(\alpha_1 \alpha_2)| |\alpha_1-\alpha_2|+|1/\alpha_2| |\beta_1-\beta_2|,
$$
we have for any $1 \leq i \leq k$,
$$
\left| \frac{p^{\circ}(z_0|z_{-n}^{-1})^{(\vec{a}_i)}}{p^{\circ}(z_0|z_{-n}^{-1})}-\frac{p^{\circ}(z_0|z_{-n-1}^{-1})^{(\vec{a}_i)}}{p^{\circ}(z_0|z_{-n-1}^{-1})} \right|
$$
$$
\leq \left| \frac{p^{\circ}(z_0|z_{-n}^{-1})^{(\vec{a}_i)}}{p^{\circ}(z_0|z_{-n}^{-1})p^{\circ}(z_0|z_{-n-1}^{-1})} \right| |p^{\circ}(z_0|z_{-n}^{-1})-p^{\circ}(z_0|z_{-n-1}^{-1})| +\left| \frac{1}{p^{\circ}(z_0|z_{-n-1}^{-1})} \right| |p^{\circ}(z_0|z_{-n}^{-1})^{(\vec{a}_i)}-p^{\circ}(z_0|z_{-n-1}^{-1})^{(\vec{a}_i)}|.
$$

It follows from multivirate Leibnitz rule and (\ref{prelim}) that there exists a positive constant $C_{\vec{a}}$ such that
\begin{equation}  \label{Austin-1}
|p(z_0|z_{-n}^{-1})^{(\vec{a})}| = |(w_{-1, -n} \Omega_{z_0} \mathbf{1})^{(\vec{a})}| \leq n^{|\vec{a}|} C_{\vec{a}}/\eps^{|\vec{a}|},
\end{equation}
and furthermore there exists a positive constant $C^{\circ}_{\vec{a}}$ such that for any $z_{-n}^{-1} \in \mathcal{Z}^n$,
\begin{equation}  \label{Austin-2}
p^{\circ}(z_0|z_{-n}^{-1})^{(\vec{a})} \leq n^{|\vec{a}|} C^{\circ}_{\vec{a}}/\eps^{|\vec{a}|+O_M}.
\end{equation}
Combining (\ref{Austin-3}), (\ref{LastFirst}), (\ref{Austin-1}) and (\ref{Austin-2}) gives us
\begin{equation} \label{G}
|(G_n^{\alpha})^{(\vec{m})}(\vec{p}, \eps)-(G_{n+1}^{\alpha})^{(\vec{m})}(\vec{p}, \eps)| = \hat{O}(\eps^n) \mbox{ on } \mathcal{M}_{\delta_0}.
\end{equation}
This implies that there exists $\eps_0 > 0$ such that $G_n^{\alpha}(\vec{p}, \eps)$ and its derivatives with respect to $(\vec{p}, \eps)$ uniformly converge on $U_{\delta_0, \eps_0}$ to a smooth function $G(\vec{p}, \eps)$ and correspondingly its derivatives (Here, by Remark~\ref{TheSameEps}, $\eps_0$ does not depend on $\vec{m}$).

\textbf{2)} It immediately follows from analyticity of $F(\vec{p}, \eps)$ and the fact that $\ord{F(\vec{p}, \eps)} \geq 1$.

\textbf{3)}
Note that,
$$
F_n(\vec{p}, \eps)-F_n^{\alpha}(\vec{p}, \eps)= \sum_{z_{-n}^{-1} \not \in T_{n}^{\alpha}, z_0} -\ord(p(z_0|z_{-n}^{-1})) p(z_{-n}^0).
$$

Apply the multivariate Leibnitz rule, then by Proposition~\ref{conditional}, (\ref{Austin-1}), (\ref{wordsINequation}) and Lemma~\ref{MATH2603}, we have for any $\ell$,
$$
\left|D_{\vec{p}, \eps}^{\ell} F_n(\vec{p}, \eps)-D_{\vec{p}, \eps}^{\ell} F_n^{\alpha}(\vec{p}, \eps) \right|=\left| \sum_{z_{-n}^{-1} \not \in T_{n}^{\alpha}, z_0} -\ord(p(z_0|z_{-n}^{-1})) D_{\vec{p}, \eps}^{\ell} (p(z_0|z_{-n}^{-1}) p(z_{-n}^{-1})) \right|=\hat{O}(\eps^n) \mbox{ on } \mathcal{M}_{\delta_0}.
$$
It follows from (\ref{FC}) and the Cauchy integral formula that
$$
\left|D_{\vec{p}, \eps}^{\ell} F_{n+1}^{\alpha}(\vec{p}, \eps)-D_{\vec{p}, \eps}^{\ell} F_n^{\alpha}(\vec{p}, \eps) \right|= \hat{O}(\eps^n) \mbox{ on } \mathcal{M}_{\delta_0},
$$
we then have
$$
\left|D_{\vec{p}, \eps}^{\ell} F_{n+1}(\vec{p}, \eps)-D_{\vec{p}, \eps}^{\ell} F_n(\vec{p}, \eps) \right|= \hat{O}(\eps^n) \mbox{ on } \mathcal{M}_{\delta_0},
$$
and thus
$$
\left|D_{\vec{p}, \eps}^{\ell} F_{n}(\vec{p}, \eps)-D_{\vec{p}, \eps}^{\ell} F(\vec{p}, \eps) \right|= \hat{O}(\eps^n) \mbox{ on } \mathcal{M}_{\delta_0},
$$
which imply that for any $\ell$, there exist $\eps_0 > 0$, $0 < \rho <1$ such that on $U_{\delta_0, \eps_0}$
$$
|D_{\vec{p},\eps}^{\ell} F_n(\vec{p}, \eps) - D_{\vec{p},\eps}^{\ell} F(\vec{p},
\eps)| < \rho^n,
$$
and further
$$
|D_{\vec{p},\eps}^{\ell} \hat{F}_n(\vec{p}, \eps) - D_{\vec{p},\eps}^{\ell} \hat{F}(\vec{p},
\eps)| < \rho^n,
$$
for sufficiently large $n$.

Similarly note that
$$
G_n(\vec{p}, \eps)-G_n^{\alpha}(\vec{p}, \eps)=\sum_{z_{-n}^{-1} \not \in T_{n}^{\alpha}, z_0} -p(z_{-n}^0) \log p^{\circ}(z_0|z_{-n}^{-1}).
$$
Then by (\ref{Austin-1}), (\ref{Austin-2}), (\ref{Austin-3}) and Lemma~\ref{MATH2603}, we have for any $\ell$,
$$
\left|D_{\vec{p}, \eps}^{\ell} G_n(\vec{p}, \eps)-D_{\vec{p}, \eps}^{\ell} G_n^{\alpha}(\vec{p}, \eps) \right|
$$
$$
=\left| \sum_{z_{-n}^{-1} \not \in T_{n}^{\alpha}, z_0} D_{\vec{p}, \eps}^{\ell}(-p(z_{-n}^{-1}) p(z_0|z_{-n}^{-1}) \log p^{\circ}(z_0|z_{-n}^{-1}))\right|=\hat{O}(\eps^n) \mbox{ on } \mathcal{M}_{\delta_0},
$$
which, together with (\ref{G}), implies that for any $\ell$, there exists $\eps_0 > 0$, $0 < \rho <1$ such that on $U_{\delta_0, \eps_0}$
$$
|D_{\vec{p},\eps}^{\ell} G_n(\vec{p}, \eps) - D_{\vec{p},\eps}^{\ell} G(\vec{p},
\eps)| < \rho^n,
$$
for sufficiently large $n$.

\end{proof}

\begin{rem}
We don't know if $G(\vec{p}, \eps)$ is analytic or not with respect to $(\vec{p}, \eps)$.
\end{rem}

\section{Concavity of Mutual Information} \label{CMI}

Recall that we are considering a parameterized family of finite-state memoryless channels with
inputs restricted to a mixing finite-type constraint $\mathcal{S}$. Again for simplicity, we assume
that $\mathcal{S}$ has order 1.

For parameter value $\eps$, the channel capacity is the supremum of the mutual information of $Z(X,\eps)$ and $X$ over all stationary input processes $X$ such that $A(X) \subseteq \mathcal{S}$. Here, we use only first order Markov input processes. While this will typically not achieve the true capacity, one can approach capacity by using Markov input processes of higher order. As in Section~\ref{asymptotics}, we identify a first order input Markov process $X$ with its joint probability vector $\vec{p} = \vec{p}_X \in \mathcal{M}$, and we write $Z = Z(\vec{p},\eps)$, thereby sometimes notationally suppressing dependence on $X$ and $\eps$.

Precisely, the {\em first order capacity} is
\begin{equation}  \label{1-order}
C^1 (\eps)=\sup_{\vec{p} \in \mathcal{M}} I(Z;X) = \sup_{\vec{p} \in \mathcal{M}} (H(Z) - H(Z|X))
\end{equation}
and its $n$-th approximation
\begin{equation}  \label{n-1-order}
C^1_n(\eps) =\sup_{\vec{p} \in \mathcal{M}} I_n(Z;X) =
\sup_{\vec{p} \in \mathcal{M}} \left( H_n(Z)-\frac{1}{n+1}H(Z_{-n}^0|X_{-n}^0) \right).
\end{equation}
As mentioned earlier, since the channel is memoryless, the second terms in (\ref{1-order}) and (\ref{n-1-order}) both reduce to $H(Z_0|X_0)$, which can be written as:
$$
\sum_{x \in \mathcal{X}, z \in \mathcal{Z}} -p(x) \sum_{c \in \mathcal{C}} p(c) p(z|x, c)  \log \sum_{c \in \mathcal{C}} p(c) p(z|x, c).
$$
Note that this expression is a linear function of $\vec{p}$ and for all $\vec{p}$ it vanishes when $\eps =0$.
Using this and the fact that for a mixing finite-type constraint there is a unique Markov chain of maximal
entropy supported on the constraint~\cite{pa64}, one can show that for sufficiently small $\eps_1 >0, \delta_1 >0$ and all $0 \le \eps \le \eps_1$,
\begin{equation}  \label{order-n}
\hspace{-0.6cm} C^1_n(\eps)=\sup_{\vec{p} \in \mathcal{M}_{\delta_1}} (H_n(Z)-H(Z_0|X_0)) > \sup_{\vec{p} \in \mathcal{M} \backslash \mathcal{M}_{\delta_1}} (H_n(Z)-H(Z_0|X_0)),
\end{equation}
\begin{equation}  \label{order-infinite}
\hspace{-0.6cm} C^1(\eps) =\sup_{\vec{p} \in \mathcal{M}_{\delta_1}} (H(Z)-H(Z_0|X_0)) >\sup_{\vec{p} \in \mathcal{M} \backslash \mathcal{M}_{\delta_1}} (H(Z)-H(Z_0|X_0)).
\end{equation}

\begin{thm}  \label{Concavity}
There exist $\eps_0 >0, \delta_0>0$ such that for all
$0 \le \eps \le \eps_0$,
\begin{enumerate}
\item
the functions $I_n(Z(\vec{p}, \eps);X(\vec{p}))$ and
$I(Z(\vec{p}, \eps);X(\vec{p}))$ are strictly concave
on $\mathcal{M}_{\delta_0}$, with unique maximizing $\vec{p}_n(\eps)$ and $\vec{p}_\infty(\eps)$;
\item
the functions $I_n(Z(\vec{p}, \eps);X(\vec{p}))$ and
$I(Z(\vec{p}, \eps);X(\vec{p}))$ uniquely achieve their maxima
on all of $\mathcal{M}$ at $\vec{p}_n(\eps)$ and $\vec{p}_\infty(\eps)$;
\item
there exists $0 < \rho <1$ such that
$$
|\vec{p}_{n}(\eps)-\vec{p}_{\infty}(\eps)| \le \rho^n.
$$
\end{enumerate}
\end{thm}

\begin{proof}

{\em Part 1:}
Recall that
$$
H(Z(\vec{p}, \eps))=G(\vec{p}, \eps)+\hat{F}(\vec{p}, \eps) (\eps \log \eps).
$$
By part 1 of Theorem~\ref{main-1}, for some  $\eps_0 >0, \delta_0>0$ , $G(\vec{p}, \eps)$ and $\hat{F}(\vec{p}, \eps)$ are smooth on $U_{\delta_0, \eps_0}$, and so
$$
\lim_{\eps \rightarrow 0} D_{\vec{p}}^2 G(\vec{p}, \eps) = D_{\vec{p}}^2 G(\vec{p}, 0)
$$
and
$$
\lim_{\eps \rightarrow 0} D_{\vec{p}}^2 \hat{F}(\vec{p}, \eps) = D_{\vec{p}}^2 \hat{F}(\vec{p}, 0),
$$
uniformly on $\vec{p} \in \mathcal{M}_{\delta_0}$. Thus,
$$
\lim_{\eps \rightarrow 0} D_{\vec{p}}^2 H(Z(\vec{p}, \eps)) = D_{\vec{p}}^2 G(\vec{p},0) = D_{\vec{p}}^2 H(Z(\vec{p},0)),
$$
again uniformly on $\mathcal{M}_{\delta_0}$. Since $D_{\vec{p}}^2 H(Z(\vec{p}, 0))$ is negative definite on $\mathcal{M}_{\delta_0}$ (see~\cite{hm06a}), it follows that for sufficiently small $\eps$, $D_{\vec{p}}^2 H(Z(\vec{p},\eps))$ is also negative definite on $\mathcal{M}_{\delta_0}$, and thus $H(Z(\vec{p},\eps))$ is also strictly concave on $\mathcal{M}_{\delta_0}$.

Since for all $\eps \ge 0$, $H(Z_0|X_0)$ is a linear function of $\vec{p}$, $I(Z(\vec{p},\eps); X(\vec{p}))$ is strictly concave on $\mathcal{M}_{\delta_0}$. This establishes part 1 for $I(Z(\vec{p}, \eps);X(\vec{p}))$. By part 2 of Theorem~\ref{main-1}, for sufficiently large $n$ ($n \ge N_1$),  we  obtain the same result (with the same $\eps_0$ and $\delta_0$) for $I_n(Z(\vec{p}, \eps);X(\vec{p}))$.  For each $1 \le n < N_1$, one can easily establish strict concavity
on $U_{\delta_n, \eps_n}$ for some $\delta_n, \eps_n >0$.

{\em Part 2:} This follows from part 1 and statements (\ref{order-n}) and (\ref{order-infinite}).

{\em Part 3:} For notational simplicity, for fixed $0 \le \eps \le \eps_0$, we rewrite $I(Z(\vec{p}, \eps); X(\vec{p})), I_n(Z(\vec{p}, \eps); X(\vec{p}))$ as function $f(\vec{p}), f_n(\vec{p})$, respectively.
By the Taylor formula with remainder, there exist $\eta_1, \eta_2 \in \mathcal{M}_{\delta_0}$ such that
$$
f(\vec{p}_n(\eps))=f(\vec{p}_{\infty}(\eps))+D_{\vec{p}}f(\vec{p}_{\infty}(\eps)) (\vec{p}_n(\eps)-\vec{p}_{\infty}(\eps))
$$
\begin{equation} \label{PANASH-1}
+(\vec{p}_n(\eps)-\vec{p}_{\infty}(\eps))^T D_{\vec{p}}^2 f(\eta_1)(\vec{p}_n(\eps)-\vec{p}_{\infty}(\eps)),
\end{equation}
$$
f_n(\vec{p}_{\infty}(\eps))=f_n(\vec{p}_{n}(\eps))+D_{\vec{p}}f_n(\vec{p}_{n}(\eps)) (\vec{p}_{\infty}(\eps)-\vec{p}_{n}(\eps))
$$
\begin{equation} \label{PANASH-2}
+(\vec{p}_n(\eps)-\vec{p}_{\infty}(\eps))^T D_{\vec{p}}^2 f_n(\eta_2)(\vec{p}_n(\eps)-\vec{p}_{\infty}(\eps)),
\end{equation}
here the superscript $T$ denotes the transpose.

By part 2 of Theorem~\ref{Concavity}
\begin{equation} \label{zero}
D_{\vec{p}}f(\vec{p}_\infty(\eps)) =0, \;\; D_{\vec{p}}f_n(\vec{p}_n(\eps)) =0.
\end{equation}
By part 2 of Theorem~\ref{main-1}, with $\ell =0$, there exists $0 < \rho_0 < 1$ such that
\begin{equation} \label{rhon}
\hspace{-0.5cm} \left|f(\vec{p}_\infty(\eps)) - f_n(\vec{p}_\infty(\eps)) \right| \le \rho_0^n,  \left|f(\vec{p}_n(\eps)) - f_n(\vec{p}_n(\eps)) \right| \le \rho_0^n.
\end{equation}
Combining (\ref{PANASH-1}), (\ref{PANASH-2}), (\ref{zero}), (\ref{rhon}), we have
$$
|(\vec{p}_n(\eps)-\vec{p}_{\infty}(\eps))^T (D_{\vec{p}}^2 f(\eta_1)+D_{\vec{p}}^2 f_n(\eta_2))(\vec{p}_n(\eps)-\vec{p}_{\infty}(\eps)) |\leq 2 \rho_0^n.
$$
Since $f$ and $f_n$ are strictly concave on $\mathcal{M}_{\delta_0}$, $D_{\vec{p}}^2 f(\eta_1), D_{\vec{p}}^2 f_n(\eta_2)$ are both negative definite. Thus there exists some positive constant $K$ such that
$$
K |\vec{p}_{n}(\eps)-\vec{p}_{\infty}(\eps)|^2 \leq 2\rho_0^n.
$$
This, together with part $1$ of Lemma~\ref{complex}, implies the existence of $\rho$.

\end{proof}

\begin{exmp}
Consider Example~\ref{rll}. For sufficiently small $\eps$ and $p$ bounded away from 0 and 1, part $1$ of Theorem~\ref{main-1} gives an expression for $H(Z(\vec{p}, \eps))$ and part $1$ of Theorem~\ref{Concavity} shows that $I(Z(\vec{p}, \eps))$ is strictly concave and thus has negative second derivative.  In this case, the
results boil down to the strict concavity of the binary entropy function; that is, when
$\eps=0$, $H(Z)=H(X)=-p \log p - (1-p) \log (1-p)$, and one computes with the second derivative with respect to $p$
$$
H''(Z)|_{\eps=0}=-\frac{1}{p}-\frac{1}{1-p} \leq -2.
$$
So, there is $\eps_0$ such that whenever $0 \le \eps \leq \eps_0$, $H''(Z) < 0$.
\end{exmp}

\end{document}